\documentclass[a4paper,twoside,reqno]{amsart}
\usepackage{amsmath}
\usepackage{amsfonts}
\usepackage{amssymb}
\usepackage{amsthm}
\usepackage{newlfont}
\usepackage{graphicx}
\usepackage{amscd}
\usepackage{young}

\theoremstyle{plain}
\newtheorem{Th}{Theorem}[section]
\newtheorem{Cor}[Th]{Corollary}

\newtheorem{Prop}[Th]{Proposition}
\theoremstyle{definition}
\newtheorem{Def}{Definition}[section]
\newtheorem{Ex}{Example}[section]
\newtheorem{Com}{Comment}[section]
\theoremstyle{remark}
\newtheorem*{Rem}{Remark}%[section]
\numberwithin{equation}{section}
\newcommand{\NN}{{\mathbb N}}

\newcommand{\RR}{{\mathbb R}}

\newcommand{\be}{\boldsymbol{e}}

\newcommand{\bs}{\boldsymbol{s}}

\begin{document}
	
	\title[Multiple orthogonal polynomials and integrable equations]
	{Determinantal approach to multiple orthogonal polynomials, and the corresponding integrable equations}
	
	\author{Adam Doliwa}
	
	\address{Faculty of Mathematics and Computer Science\\
		University of Warmia and Mazury in Olsztyn\\
		ul.~S{\l}oneczna~54\\ 10-710~Olsztyn\\ Poland} 
	\email{doliwa@matman.uwm.edu.pl}
	\urladdr{http://wmii.uwm.edu.pl/~doliwa/}
	
	%
%	\date{\today}
	
	\begin{abstract}
We study multiple orthogonal polynomials exploiting their explicit determinantal representation in terms of moments.  Our reasoning follows that applied to solve the Hermite--Pad\'{e} approximation and interpolation problems. We study also families of multiple orthogonal polynomials obtained by variation of the measures known from the theory of discrete-time Toda lattice equations. We present determinantal proofs of certain fundamental results of the theory, obtained earlier by other authors in a different setting. We derive also quadratic identities satisfied by the polynomials, which are new elements of the theory. Resulting equations allow to present multiple orthogonal polynomials within the theory of integrable systems.
	\end{abstract}
	\maketitle
	
	\section{Introduction}
Multiple orthogonal polynomials \cite{Aptekarev,NikishinSorokin,MF-VA} are a generalization of orthogonal polynomials in which the orthogonality is distributed among a number of orthogonality
weights. The standard orthogonal polynomials~\cite{Chihara,Nevai} play a distinguished role in theoretical physics and applied mathematics. They can be encountered when solving, by separation of variables, partial differential equations of the classical field theory or the quantum mechanics~\cite{Nikiforov-Suslov-Uvarov}.
Another source of their applications are problems of pro\-ba\-bi\-li\-ty theory in the context of random walks \cite{LedermannReuter,KarlinMcGregor,Schoutens,Iglesia}; see also \cite{CMGV-2,DoliwaSiemaszko-QW} for generalization to quantum walks. By relation to continued fractions the orthogonal polynomials appear in number theory and combinatorics~\cite{Flajolet,Viennot}.
They are important in the spectral theory of operators in Hilbert space~\cite{Akhiezer,Szego,Geronimus}, and play prominent role in the modern theory of representations of groups and algebras~\cite{Vilenkin-Klimyk}, including the quantum groups~\cite{Klimyk-Schmudgen}. More recently, orthogonal polynomials found application in theory of random matrices~\cite{AdlervanMoerbekeVanhaecke,Deift} and were used to construct special solutions to Painlev\'{e} equations~\cite{Clarkson,VanAssche}. The structural affinity to Pad\'{e} approximants~\cite{Brezinski,Baker,Gragg} and relation to the Toda lattice equations~\cite{Toda-TL} broaden the theory and applications of orthogonal polynomials. Especially that last point allows to consider orthogonal polynomials as a part of the theory of integrable systems~\cite{AdlervanMoerbeke,IDS} and use its powerful techniques, like for example the Riemann-Hilbert problem~\cite{AdlervanMoerbeke,Deift}.
It is therefore plausible to expect that multiple orthogonal polynomials will enlarge the range of applications of orthogonal polynomials to more advanced problems. Indeed, the progress that has been made in recent years both in understanding their theory and in finding interesting applications seems to confirm this opinion. 

The development of the theory o multiple orthogonal polynomials in XXth century is summarized in~\cite{NikishinSorokin,Aptekarev}. Properties of multiple orthogonal polynomials for classical weights were described in~\cite{VanAsscheCoussement,AptekarevBranquinhoVanAssche}. Application to random matrices with external source and non-intersecting path ensembles was the subject of \cite{BleherKuijlaars,DaemsKuilaars,Kuijlaars}. Also some structural results generalizing those of orthogonal polynomials can be found in~\cite{CoussementVanAasche,Ismail,VanAsche-nn-mop,BranquinhoMorenoManas}. Among them the relation~\cite{Alvarez-FernandezPrietoManas,AptekarevDerevyaginMikiVanAssche,FernandezManas} to integrable equations is of particular interest.

One of possible paths to understand integrability is that initiated by Hirota, and developed by Sato and his school \cite{Sato,DKJM}. On the elementary level it is based on determinantal identities~\cite{Hirota-book} satisfied by the so called $\tau$-function; see \cite{Harnad} for modern presentation of its advanced aspects. It is commonly accepted that one of the most important integrable discrete systems \cite{IDS} both from theoretical point of view and of the application is Hirota's discrete Kadomtsev--Petviashvili equation~\cite{Hirota}. It contains, as symmetry reductions or appropriate limits, many known integrable equations, notably the whole  Kadomtsev--Petviashvili hierarchy~\cite{DKJM}. See also \cite{Miwa,Shiota,Zabrodin,KNS-rev,Nimmo-NCKP, BialeckiDoliwa,Dol-Des,Dol-AN} for other aspects of the Hirota system both on the classical and quantum levels. 

The goal of the paper is to present the theory of (multiple) orthogonal polynomials as a part of the theory of integrable systems. Relation to integrability of the close connected subject~\cite{NikishinSorokin,Aptekarev,AptekarevKuijlaars,VanAsche-HPAO} of Hermite--Pad\'{e} (type I) approximation~\cite{Baker} has been clarified recently in~\cite{Doliwa-Siemaszko-2}, where determinantal identities were used to formulate corresponding reduction of the Hirota system on the level of $\tau$-function. 
Notice that integrable equations appear in the theory of orthogonal polynomials upon introducing time variable by certain evolution of the measure~\cite{Moser,Flaschka}, and the Painlev\'{e} equations need special forms of the measure \cite{Sogo,VanAssche}. In the corresponding Pad\'{e} approximation and interpolation techniques \cite{Yamada,Nagao-Yamada,Nagao} all variables appear on equal footing. The integrability features of the interpolation analog of the Hermite--Pad\'{e} problem were presented in~\cite{Doliwa-NAMTS}.

It is well known~\cite{NikishinSorokin,Aptekarev} that the multiple orthogonal polynomials give rise to a particular simultaneous Pad\'{e} (or the Hermite--Pad\'{e} of type II) approximation problem. The determinantal solution of generic version of the problem is given in~\cite{BeckermannLabahn}, and in principle it can be used to study multiple orthogonality. Equivalently, one can use Mahler's duality~\cite{Mahler-P} to transfer results obtained in \cite{Doliwa-Siemaszko-2} for the corresponding Hermite--Pad\'{e} type~I problem. However we found it simpler and more efficient to build directly the determinantal approach for multiple orthogonal polynomials without the need to go around via the Hermite--Pad\'{e} problems, especially that it has been already initiated in~\cite{Aptekarev,NikishinSorokin,AptekarevDerevyaginMikiVanAssche}.  We are going to expand their work, in particular we provide the direct determinantal proofs of fundamental equations of the theory obtained there on a different basis. Moreover we give certain quadratic identities satisfied by the multiple orthogonal polynomials which are new in this context, although their analogs for solutions of the  Hermite--Pad\'{e} approximation and interpolation problems were given recently in~\cite{Doliwa-Siemaszko-2,Doliwa-NAMTS}.

The structure of the paper is as follows. In Section~\ref{sec:multiple-OP}  we recall basic elements of the theory of multiple orthogonal polynomials including their determinantal representation. Already at this level we encounter fundamental integrable equations, what seems not to be noticed before. In particular, we show that the polynomials themselves satisfy Hirota's quadratic identities. In the next Section we discuss determinantal derivation of the corresponding generalization of the three term recurrence formula and its compatibility condition. Then in Section~\ref{sec:MTS} we study, following simplified version of the approach initiated in~\cite{AptekarevDerevyaginMikiVanAssche}, a particular discrete-time Toda type evolution of the measures. We give determinantal derivation of the additional equations, which express the corresponding evolution of the polynomials and related functions. The resulting equations will be discussed in the next Section~\ref{sec:MOP-I} within the more general context of integrable systems theory. 

\section{Multiple orthogonal polynomials and their simplest properties} 
\label{sec:multiple-OP} 
In this Section we consider a generalization~\cite{Aptekarev,NikishinSorokin} of the theory of orthogonal polynomials, which has found recently application in mathematical physics and theory of random processes~\cite{AptekarevDerevyaginMikiVanAssche,BleherKuijlaars,Kuijlaars,VanAsche-HPAO,VanAsche-nn-mop}. We assume that the Reader knows basic properties of the standard orthogonal polynomials on the level of first sections of textbooks~\cite{Chihara,Ismail}. Also some knowledge about their relation to integrable systems and Pad\'{e} approximants, as described for example in \cite{IDS}, may be useful to provide additional motivation.  In the paper we do not present specific examples of families of multiple orthogonal polynomials, in particular those which generalize the most popular standard orthogonal polynomials, as they and their applications and particular properties can be found in the existing literature~\cite{Aptekarev,AptekarevBranquinhoVanAssche,BleherKuijlaars,CoussementVanAasche,DaemsKuilaars,FernandezManas,FilipukVanAsscheZhang,Kuijlaars,NikishinSorokin,VanAsscheCoussement}.

\subsection{Basic formulas}

\begin{Def}
Given $r$ positive nontrivial measures $\mu^{(j)}$, $j=1,\dots , r$, on the real line, and given $r$ non-negative integers $\bs = (s_1,\dots , s_r) \in\mathbb{N}_0^r$. The \emph{multiple orthogonal polynomial of the index $\bs$} with respect to the measures is a polynomial $Q_{\bs}(x)$  of degree $|\bs|=s_1 + \dots + s_r$ with real coefficients,
which satisfies the following orthogonality conditions with respect to the measures
\begin{equation} \label{eq:orth-r}
\int_\mathbb{R} x^k  Q_{\bs}(x) d\mu^{(j)}(x) = 0, \qquad k=0,1,\dots , s_j - 1 .
\end{equation}
\end{Def}
Conditions \eqref{eq:orth-r} give $|\bs|$ linear equations for $|\bs|+1$ coefficients of the polynomial. Notice that when $s_j = 0$ then there is no equation coming from the measure $\mu^{(j)}$. 
To present a solution of the system let us consider the Cauchy--Stieltjes transforms of the measures 
\begin{equation} 
S^{(j)}(z) = \int_\mathbb{R} \frac{d\mu^{(j)}(x)}{z-x} =
\sum_{k=0}^\infty \frac{\nu_{k}^{(j)}}{z^{k+1}} , \qquad z \to \infty,
\end{equation}
where the moments are given by
\begin{equation} 
 \nu_{k}^{(j)} = \int_{\mathbb{R}} x^k d\mu^{(j)}(x).
\end{equation}
Assume that all the moments exist then the matrix $M_{\bs}$ of the aforementioned linear equations, of $|\bs|$ rows and $|\bs| + 1$ columns,  reads
\begin{equation} \label{eq:M-r}
M_{\bs} = \begin{pmatrix}
\nu_{0}^{(1)} & \nu_{1}^{(1)}  &  \dots & \nu_{|\bs|}^{(1)} \\
\vdots & \vdots & \ddots   & \vdots \\
\nu_{s_1 - 1}^{(1)}  & \nu_{s_1}^{(1)}  &  \dots & \nu_{ |\bs|+ s_1 - 1}^{(1)}  \\
\cdots & \cdots &   & \cdots \\
\nu_{0}^{(r)} & \nu_{1}^{(r)} &  \dots & \nu_{|\bs|}^{(r)}\\
\vdots & \vdots & \ddots   & \vdots \\
\nu_{s_r - 1}^{(r)} & \nu_{s_r}^{(r)} &  \dots & \nu_{ |\bs| + s_r - 1}^{(r)} 
\end{pmatrix} .
\end{equation}	
It is composed out of $r$ row-blocks, the $j$th block consists of $s_j$ rows obtained from the orthogonality conditions~\eqref{eq:orth-r}. We will use also its column description
\begin{equation}
M_{\bs} = \left( \begin{array}{cccc} \boldsymbol{\nu}_{\bs ,0} & \boldsymbol{\nu}_{\bs ,1} & \dots & \boldsymbol{\nu}_{\bs ,|\bs|} \end{array} \right),
\end{equation}
where definition of each column can be easily deduced from \eqref{eq:M-r}.

Define the polynomial $Z_{\bs}(x)$  by the determinant
\begin{equation} \label{eq:Z-r}
Z_{\bs}(x) = \left| \begin{matrix}
\nu_{0}^{(1)} & \nu_{1}^{(1)}  &  \dots & \nu_{|\bs|}^{(1)} \\
\vdots & \vdots & \ddots   & \vdots \\
\nu_{s_1 - 1}^{(1)}  & \nu_{s_1}^{(1)}  &  \dots & \nu_{ |\bs|+ s_1 - 1}^{(1)}  \\
\cdots & \cdots &   & \cdots \\
\nu_{0}^{(r)} & \nu_{1}^{(r)} &  \dots & \nu_{|\bs|}^{(r)}\\
\vdots & \vdots & \ddots   & \vdots \\
\nu_{s_r - 1}^{(r)} & \nu_{s_r}^{(r)} &  \dots & \nu_{ |\bs| + s_r - 1}^{(r)} \\
1 & x & \dots & x^{|\bs|}
\end{matrix} \right| =
 \left| \begin{array}{ccccc} \boldsymbol{\nu}_{\bs,0} &  \boldsymbol{\nu}_{\bs,1} & \dots & \boldsymbol{\nu}_{\bs,|\bs|} \\
1 & x & \dots & x^{|\bs|}
\end{array}
\right|,
\end{equation}	
of the matrix $M_{\bs}$ supplemented by the row
$\left(  1 , x , x^2 , \dots , x^{|\bs|}  \right)$.
It satisfies equation~\eqref{eq:orth-r}, as then two rows of the resulting determinant on the left hand side coincide. The polynomial $Z_{\bs}(x)$  will be called the \emph{canonical multiple orthogonal polynomial} of multi-index $\bs$ for measures~ $\mu^{(j)}$, $j=1,\dots , r$. 

When the determinant
\begin{equation} \label{eq:D-r}
\qquad \quad
D_{\bs}  = \left| \begin{matrix}
\nu_{0}^{(1)} & \nu_{1}^{(1)} &  \dots & \nu_{|\bs|-1}^{(1)}\\
\vdots & \vdots & \ddots   & \vdots \\
\nu_{s_1 - 1}^{(1)} & \nu_{s_1}^{(1)} &  \dots & \nu_{ |\bs|+ s_1 - 2}^{(1)} \\
\cdots & \cdots &   & \cdots \\
\nu_{0}^{(r)} & \nu_{1}^{(r)} &  \dots & \nu_{|\bs|-1}^{(r)}\\
\vdots & \vdots & \ddots   & \vdots \\
\nu_{s_r - 1}^{(r)} & \nu_{s_r}^{(r)} &  \dots & \nu_{ |\bs| + s_r - 2}^{(r)} 
\end{matrix} \right| = \left|\begin{array}{cccc}\boldsymbol{\nu}_{\bs ,0} & \boldsymbol{\nu}_{\bs ,1} & \dots & \boldsymbol{\nu}_{\bs ,|\bs|-1} \end{array}\right|,
\end{equation}
does not vanish, then multiple orthogonal polynomials of the multi-index $\bs$ form one dimensional family, and the corresponding multi-index is called \emph{normal}. In such case usually one takes for the solution the monic polynomial \cite{Aptekarev} 
\begin{equation} \label{eq:QZD}
Q_{\bs}(x) = Z_{\bs}(x) / D_{\bs}.
\end{equation}

\subsection{Simplest equations}
We start investigation of equations satisfied by multiple orthogonal polynomials by showing that with respect to the discrete indices $\bs = (s_1, s_2 \dots, s_r) \in \NN_0^r$ they provide solutions to Hirota's discrete Kadomtsev--Petviashvili system and of its linear problem. In next parts of the paper we discuss additional equations satisfied by the polynomials. Below, and in all the paper, by $\be_j$, $j=1,\dots , r$, we denote the canonical vector of dimension $r$ with all its components being zeros except of one at $j$th place.

\begin{Rem}
In deriving the equations we will use two determinantal identities~\cite{BrualdiSchneider}. Both relate the determinant $A=\det(a_{ij})_{i,j = 1,\dots , n}$ to determinants of the matrix with certain rows and columns removed. Choose $k$ row indices $i_1 < i_2 < \dots < i_k$ from $1,2,3,\dots , n$, and similarly  consider columns $j_1 < j_2 < \dots < j_k$ from $1,2,3,\dots , n$. By $A\begin{pmatrix}
i_1, \dots , i_k\\ j_1, \dots , j_k\end{pmatrix}$
denote  determinant of the matrix obtained by removing all rows other then $i_1, \dots , i_k$ and all columns other then $j_1,\dots , j_k$.  Determinant of the matrix obtained by removing all rows $i_1, \dots , i_k$ and all columns $j_1,\dots , j_k$ denote by $A(i_1, \dots , i_k|j_1,\dots , j_k)$. Then by the generalized Laplace expansion with respect to the $k$ rows $i_1, \dots , i_k$ we have
\begin{equation*}
A = \sum_{1\leq j_1 < \dots < j_k \leq n} (-1)^{i_1 + i_2 + \dots + i_k + j_1 + j_2 + \dots + j_k} A\begin{pmatrix}
i_1, \dots , i_k\\ j_1, \dots , j_k\end{pmatrix}
A(i_1, \dots , i_k|j_1,\dots , j_k),
\end{equation*}
where the sum is over all possible choices of $k$ columns.
Analogously one defines generalized Laplace expansion with respect to $k$ columns.
We will also use the following identity, which is a consequence of the generalized Laplace expansion, attributed to Jacobi and Sylvester,
\begin{equation*}
A \, A( i_1, i_2|j_1, j_2) =
A(i_1|j_1)\, A(i_2|j_2) - 
A(i_1|j_2)\, A(i_2|j_1) , 
\end{equation*}
and known also as the Dodgson determinant condensation rule. 
\end{Rem}

Let us first consider the following relation between the polynomials with neighboring multi-indices, which within the context of integrable systems provides the so called linear problem for the Hirota equations.
\begin{Prop}
	When $r\geq 2$ then the canonical multiple polynomials $Z_{\bs}(x)$ and the determinants $D_{\bs}$ satisfy, as functions of the multi-index $\bs$, the following system
\begin{equation} \label{eq:Z-D-jk}
Z_{\bs+\be_j}(x) D_{\bs + \be_k} - Z_{\bs+\be_k}(x) D_{\bs + \be_j}=  Z_{\bs}(x) D_{\bs+\be_j + \be_k}, \qquad j<k.
\end{equation}
\end{Prop}
\begin{proof}
Apply the Sylvester identity to the determinant
	\begin{equation*}	
	\left| \begin{array}{cccccc} \boldsymbol{\nu}_{\bs+\be_j + \be_k,0} &  \boldsymbol{\nu}_{\bs+\be_j + \be_k,1} & \dots & \boldsymbol{\nu}_{\bs+\be_j + \be_k,|\bs|+1} & \boldsymbol{0} \\
	1 & x & \dots & x^{|\bs|+1} & 1
	\end{array}
	\right|,
	\end{equation*}
the last rows of the $j$th and $k$th row-blocks, and its two last columns.
\end{proof}

\begin{Cor}
	When the determinants $D_{\bs}$, $D_{\bs + \be_j}$ and $D_{\bs + \be_k}$ do not vanish then equation~\eqref{eq:Z-D-jk} gives	
	\begin{equation} \label{eq:Q-C-jk}
	Q_{\bs+\be_j}(x) - Q_{\bs+\be_k}(x)= C_{\bs}^{(jk)}  Q_{\bs}(x) , \qquad j<k,
	\end{equation}
	where \begin{equation} \label{eq:C-jk}
	C_{\bs}^{(jk)} =  \frac{D_{\bs} D_{\bs+\be_j + \be_k}}{  D_{\bs + \be_j}  D_{\bs + \be_k} } = - C^{(kj)}_{\bs}, \qquad j<k.
	\end{equation}
\end{Cor}

The following relation between the determinants $D_{\bs}$ is known in the theory of integrable systems as Hirota's discrete Kadomtsev--Petviashvili equation.
\begin{Prop}
When $r\geq 3$ then the determinants $D_{\bs}$ satisfy the non-linear Hirota equations
\begin{equation} \label{eq:D-ijk-H}
D_{\bs+\be_i + \be_j} 	D_{\bs+\be_k} -	D_{\bs+ \be_i + \be_k} 	D_{\bs+\be_j} + 	D_{\bs+\be_j + \be_k} 	D_{\bs+\be_i} = 0, \qquad i < j < k.
\end{equation}
\end{Prop}
\begin{proof}
Consider the matrix whose determinant gives $D_{\bs +\be_i + \be_j + \be_k}$, but replace its last column by the vector $(0, 0, \dots ,0,1,0,\dots, 0)^T$ consisting of zeros except of $1$ in the last row of $i$th block. Then apply to such new matrix the Sylvester identity with the last rows of the $j$th and $k$th row-blocks, and with its two last columns.
\end{proof}
\begin{Com}
	Our proof applies also to non-normal multi-indices. The standard derivation of \eqref{eq:D-ijk-H}, in the context of integrable systems, starts with three copies of \eqref{eq:Z-D-jk} for pairs $i<j$, $i<k$, $j<k$, with the first equation multiplied by $D_{\bs + \be_k}$, the second multiplied by $D_{\bs + \be_j}$, and the third one by $D_{\bs + \be_i}$. 
\end{Com}

In equations~\eqref{eq:Z-D-jk} the polynomials $Z_{\bs}(x)$ can be exchanged with the functions $D_{\bs}$, provided shifts in positive directions are simultaneously replaced by the shifts backward. Such \emph{duality}~\cite{Saito-Saitoh} implies the following result, which in the context of multiple orthogonal polynomials seems not to be considered before.
\begin{Prop}
	When $r\geq 3$ then the polynomials $Z_{\bs}(x)$ satisfy, with respect to discrete multi-indices, the non-linear Hirota equations
\begin{equation}
\label{eq:Z-ijk-H}
Z_{\bs+\be_i + \be_j}(x) Z_{\bs+\be_k}(x) -	Z_{\bs+ \be_i + \be_k}(x) 	Z_{\bs+\be_j}(x) + 	Z_{\bs+\be_j + \be_k}(x) 	Z_{\bs+\be_i}(x) = 0, \quad i<j<k.
\end{equation}
\end{Prop}
\begin{proof}
Consider the matrix whose determinant gives $Z_{\bs +\be_i + \be_j + \be_k}(x)$, but replace its last column by the vector $(0, 0, \dots ,0,1,0,\dots, 0)^T$ consisting of zeros except of $1$ in  the last row of $i$th block. Then apply the Sylvester identity with the last rows of the $j$th and $k$th row-blocks, and with its two last columns.
\end{proof}
\begin{Com}
Equation \eqref{eq:Z-ijk-H} can be again demonstrated using three copies of the linear problem \eqref{eq:Z-D-jk} for pairs $i<j$, $i<k$ and $j<k$. It is convenient to shift the first copy in direction $\be_k$ and multiply by $Z_{\bs + \be_i + \be_j}$. Analogously one does for two other copies, and then subtract the second equation from the sum of other two.
\end{Com}	

\begin{Rem}
	Because $Z_{\bs}(x) = D_{\bs} x^{|\bs|} + \text{lower order terms}$, then the Hirota equations \eqref{eq:D-ijk-H} can be derived from \eqref{eq:Z-ijk-H} by collecting the highest order term coefficients. 
\end{Rem}

\begin{Rem} 
For completeness we write down another similar identities holding for $r\geq 3$
	\begin{equation} 	\label{eq:Z-D-ijk-H2}
		D_{\bs+\be_i + \be_j} Z_{\bs+\be_k}(x) -	D_{\bs+ \be_i + \be_k}	Z_{\bs+\be_j}(x) + 	D_{\bs+\be_j + \be_k}	Z_{\bs+\be_i}(x) = 0, \qquad i<j<k.
	\end{equation}
They can be demonstrated using three copies of the linear problem \eqref{eq:Z-D-jk} for pairs $i<j$, $i<k$ and $j<k$ or by appropriate Sylvester identities. Details are left for the Reader as a simple exercise.
\end{Rem}	

\section{Generalization of the three-term relation}

\label{sec:G3T}

The celebrated three-term relation \cite{Chihara,Ismail} for orthogonal polynomials links them to the spectral theory of operators~\cite{Akhiezer}, random walks~\cite{KarlinMcGregor} or integrable systems~\cite{Flaschka,IDS}. In this Section we study its generalization~\cite{VanAsche-nn-mop} to multiple orthogonal polynomials, together with related questions, using determinantal tools.

\subsection{Determinantal derivation of the multi-term relations}
For multiple orthogonal polynomials we have the following analog~\cite{VanAsche-nn-mop} of the three-term relation
\begin{equation} \label{eq:lin-term-r}
x Q_{\bs} (x) = Q_{\bs + \be_j}(x) + b_{\bs}^{(j)} Q_{\bs}(x) + \sum_{i=1}^r a_{\bs}^{(i)} Q_{\bs - \be_i}(x), \qquad j=1,2,\dots , r,
\end{equation}
where the recurrence coefficients are given by the formulas
\begin{equation} \label{eq:a-b-r}
a_{\bs}^{(j)} =  \frac{D_{\bs + \be_j} D_{\bs - \be_j}}{D_{\bs}^2},\qquad 
b_{\bs}^{(j)} = \frac{\tilde{D}_{\bs + \be_j}}{D_{\bs + \be_j}} - \frac{\tilde{D}_{\bs}}{D_{\bs}} ,
\end{equation}
and
\begin{equation} \label{eq:tD}
\tilde{D}_{\bs} = \left|\begin{array}{ccccc} \boldsymbol{\nu}_{\bs ,0} & \boldsymbol{\nu}_{\bs ,1} & \dots & \boldsymbol{\nu}_{\bs ,|\bs|-2}& \boldsymbol{\nu}_{\bs ,|\bs|} \end{array}
\right|,
\end{equation} 
is determinant of the matrix obtained by removing the penultimate column from $M_{\bs}$. In the standard derivation 
\cite{Ismail} of the multi-term relations \eqref{eq:lin-term-r} the orthogonality conditions \eqref{eq:orth-r} and the expansion 
\begin{equation} \label{eq:Q-exp-r}
Q_{\bs}(x) =  x^{|\bs|} - \frac{\tilde{D}_{\bs}}{D_{\bs}} x^{|\bs|-1} + \dots,
\end{equation}
are used.

Application of the Sylvester identity to $D_{\bs+\be_j + \be_k}$, $j<k$, with respect to its two last columns and the last rows of the $j$th and $k$th row-blocks gives
	\begin{equation} \label{eq:Djk-tD}
	D_{\bs+\be_j + \be_k} D_{\bs} = \tilde{D}_{\bs + \be_k} D_{\bs + \be_j} - \tilde{D}_{\bs + \be_j} D_{\bs + \be_k}.
	\end{equation}
Equivalently, the above system can be obtained by collecting coefficients in front of $x^{|\bs|}$ in equations~\eqref{eq:Z-D-jk}.

Our goal of this Section will be to show determinantal meaning of \eqref{eq:lin-term-r}, so we rewrite it in more convenient  form for this purpose. 
\begin{Prop} \label{prop-r-term}
	The canonical multiple orthogonal polynomials satisfy the following form of the multi-term relation \eqref{eq:lin-term-r}
\begin{equation} \label{eq:3-term-r-Z}
x Z_{\bs}(x) {D_{\bs}} = \frac{{D_{\bs}}}{D_{\bs + \be_j} }  \Bigl( Z_{\bs + \be_j}(x) D_{\bs} + Z_{\bs}(x) \tilde{D}_{\bs + \be_j}\Bigr)  - \tilde{D}_{\bs} Z_{\bs}(x) + \sum_{i=1}^r D_{\bs + \be_i} Z_{\bs - \be_i}(x) ,
\end{equation} 	
where $j$ is an arbitrary index from $1,2,\dots ,r$.
\end{Prop}
\begin{proof}
As we mentioned above, equation \eqref{eq:3-term-r-Z} is just reformulation of \eqref{eq:lin-term-r} using determinantal definitions~\eqref{eq:QZD} and~\eqref{eq:a-b-r}. We will however re-derive it in the spirit of the present paper. 
The Sylvester identity applied to the determinant $Z_{\bs+\be_j}(x)$, its two last columns, the ultimate row of the $j$th row-block and the last row gives
\begin{equation} \label{eq:ZD-tt}
Z_{\bs+\be_j}(x) D_{\bs} = \tilde{Z}_{\bs}(x) D_{\bs+\be_j} - Z_{\bs}(x) \tilde{D}_{\bs + \be_j},
\qquad j=1,2,\dots , r,
\end{equation}
where $\tilde{Z}_{\bs}(x)$ is defined in analogy to $\tilde{D}_{\bs}$ as
\begin{equation} \label{eq:tZ}
\tilde{Z}_{\bs}(x)=
\left| \begin{array}{cccccc} \boldsymbol{\nu}_{\bs,0} &  \boldsymbol{\nu}_{\bs,1} & \dots & \boldsymbol{\nu}_{\bs,|\bs|-1}& \boldsymbol{\nu}_{\bs,|\bs|+1} \\
1 & x & \dots & x^{|\bs|-1}& x^{|\bs|+1} 
\end{array}\right|.
\end{equation}

Consider the following determinant of order $(2|\bs|+1)$
\begin{equation} \label{eq:det-lin-Z-D}
\mathcal{N}_{\bs}(x) = \left| \begin{array}{ccccccccc}
\boldsymbol{\nu}_{\bs,0} &  \boldsymbol{\nu}_{\bs,1} & \dots &  \boldsymbol{\nu}_{\bs, |\bs|-1 } &  \boldsymbol{\nu}_{\bs, |\bs| } & & & \boldsymbol{0}^{|\bs|}_{|\bs|}  &  \\
\boldsymbol{\nu}_{\bs,1} &  \boldsymbol{\nu}_{\bs,2} & \dots &  \boldsymbol{\nu}_{\bs, |\bs| } &  \boldsymbol{\nu}_{\bs, |\bs| + 1} &  \boldsymbol{\nu}_{\bs,0} &  \boldsymbol{\nu}_{\bs,1} & \dots &  \boldsymbol{\nu}_{\bs, |\bs| -1 }
\\
x & x^2 &  \dots & x^{|\bs|} & x^{|\bs|+1} & 1 & x & \dots & x^{|\bs|-1}
\end{array}
\right| ,
\end{equation}
where by $\boldsymbol{0}^a_b$ we denote block of zeros with $a$ rows and $b$ columns. Subtract the last $|\bs|-1$ columns from the first ones, what gives
\begin{equation*} \mathcal{N}_{\bs}(x) = 
\left| \begin{array}{cccccccc}
\boldsymbol{\nu}_{\bs,0} &  \dots & \boldsymbol{\nu}_{\bs,|\bs|-2}&  \boldsymbol{\nu}_{\bs, |\bs|-1 } &  \boldsymbol{\nu}_{\bs, |\bs| } & & \boldsymbol{0}^{|\bs|}_{|\bs|} &  \\ &&& 
\boldsymbol{\nu}_{\bs,|\bs|} &  \boldsymbol{\nu}_{\bs,|\bs|+1} &  \boldsymbol{\nu}_{\bs,0} & \dots &  \boldsymbol{\nu}_{\bs, |\bs|-1 } 
\\
&  \boldsymbol{0}^{|\bs|+1}_{|\bs|-1} & &
x^{|\bs|} & x^{|\bs|+1} & 1 &  \dots & x^{|\bs|-1} 
\end{array}
\right| , 
\end{equation*}
and by the generalized Laplace expansion with respect to the first $|\bs|$ rows leads to
\begin{equation} \label{eq:ZD-1}
\mathcal{N}_{\bs}(x) = (-1)^{|\bs|}\bigl( D_{\bs} \tilde{Z}_{\bs}(x) - \tilde{D}_{\bs} Z_{\bs}(x)\bigr).
\end{equation}

On the other side, row operations using first $|\bs|$ rows, applied to the  determinant \eqref{eq:det-lin-Z-D} allow to remove most of rows of the left part of the next $|\bs|$ rows,  except of the last rows of each of the $r$ row-blocks. Those are needed to complete the upper left block of size $|\bs| \times (|\bs| + 1)$ to $D_{\bs + \be_i}$. Then the generalized Laplace expansion with respect the first $(|\bs |+1)$ columns, where each term has to use the first $|\bs|$ rows, gives the sum of $(-1)^{|\bs|} xZ_{\bs}(x) D_{\bs}$ (it uses the last row) and, up to signs, terms of the form $D_{\bs + \be_i} Z_{\bs - \be_i}(x)$. Careful examination of the signs, which are due to both the generalized Laplace expansion formula and the transpositions needed, gives 
\begin{equation} \label{eq:ZD-2}
\mathcal{N}_{\bs}(x) = (-1)^{|\bs|} \Bigl( x Z_{\bs}(x) D_{\bs} - \sum_{i=1}^r D_{\bs + \be_i} Z_{\bs - \be_i}(x)\Bigr).
\end{equation}
Comparison of both expressions \eqref{eq:ZD-1} and \eqref{eq:ZD-2} gives equation
\begin{equation} \label{eq:xZD}
x Z_{\bs}(x) D_{\bs} = \tilde{Z}_{\bs}(x) D_{\bs}  -\tilde{D}_{\bs} Z_{\bs}(x) + \sum_{i=1}^r D_{\bs + \be_i} Z_{\bs - \be_i}(x) .
\end{equation} 
By excluding polynomial $\tilde{Z}_{\bs}(x)$ from the above equation, with the help of identity \eqref{eq:ZD-tt}, we obtain formula \eqref{eq:3-term-r-Z}.
\end{proof}
\begin{Rem}
The crucial part of the proof, i.e.  equation~\eqref{eq:xZD}, is based on calculation of a given determinant in two ways: (i)~we apply generalized row Laplace expansion preceded by certain column operations, (ii) we apply generalized column Laplace expansion preceded by row operations. Analogous patterns of reasoning in proving similar formulas will appear several times in this work.
\end{Rem}

\subsection{Compatibility conditions}
In contrary to the single measure case $r=1$ the coefficients $a^{(j)}_{\bs}$ and $b^{(j)}_{\bs}$ in the system of equations~\eqref{eq:lin-term-r} cannot be given arbitrarily but should satisfy~\cite{VanAsche-nn-mop}  compatibility conditions of the system
	\begin{align} \label{eq:mo-1}
	\quad \; b_{\bs+\be_k}^{(j)} - b_{\bs + \be_j}^{(k)} & = b_{\bs}^{(j)} - b_{\bs}^{(k)},\\  \label{eq:mo-2}
	b_{\bs}^{(k)} b_{\bs + \be_k}^{(j)} - b_{\bs}^{(j)} b_{\bs + \be_j}^{(k)}  & = \sum_{i=1}^r \left( a_{\bs+\be_k}^{(i)} - a_{\bs+\be_j}^{(i)}\right),
	\qquad j\neq k,\\  \label{eq:mo-3}
	a_{\bs + \be_k}^{(j)}( b_{\bs -\be_j}^{(j)} - b_{\bs - \be_j}^{(k)}) & = a_{\bs}^{(j)}(b_{\bs}^{(j)} - b_{\bs}^{(k)}).
	\end{align}
One can check that equation \eqref{eq:mo-1} is immediate consequence of  definition \eqref{eq:a-b-r} of the coefficients, while derivation of \eqref{eq:mo-3} involves in addition also the Sylvester identity~\eqref{eq:Djk-tD}. Determinantal verification of equations~\eqref{eq:mo-2} is more laborious. To make the paper complete, even if the identity is a direct consequence of the previous considerations, we present its derivation below as an exercise in determinantal calculus. The calculation provides also additional motivation to introduce certain functions in the next parts of the Section. 
\begin{Ex} \label{ex:t-d}
Application of the method used in proof of Proposition~\ref{prop-r-term} to the following determinant of order $2|\bs|$	
\begin{equation*} 
\left| \begin{array}{ccccccccc}
\boldsymbol{\nu}_{\bs,0} &  \boldsymbol{\nu}_{\bs,1} & \dots &  \boldsymbol{\nu}_{\bs, |\bs|-1 } &  \boldsymbol{\nu}_{\bs, |\bs|} & & & \boldsymbol{0}^{|\bs|}_{|\bs|-1}  &  \\ \boldsymbol{\nu}_{\bs,1} &  \boldsymbol{\nu}_{\bs,2} & \dots &  \boldsymbol{\nu}_{\bs, |\bs|  } &  \boldsymbol{\nu}_{\bs, |\bs| +1 } &  \boldsymbol{\nu}_{\bs,0} &  \boldsymbol{\nu}_{\bs,1} & \dots &  \boldsymbol{\nu}_{\bs, |\bs| -2}
\end{array}
\right|	,
\end{equation*}
leads to 
\begin{equation}  \label{eq:mo2-p-1}
 \sum_{i=1}^r D_{\bs+\be_i} D_{\bs - \be_i} = D_{\bs} \hat{\tilde{D}}_{\bs} - \tilde{D}^2_{\bs} + D_{\bs} \hat{D}_{\bs},
\end{equation}
where
\begin{align} \label{eq:h-D}
\hat{D}_{\bs} & = \left|\begin{array}{cccccc} \boldsymbol{\nu}_{\bs ,0} &  \boldsymbol{\nu}_{\bs ,1} &\dots & \boldsymbol{\nu}_{\bs ,|\bs|-3} & \boldsymbol{\nu}_{\bs ,|\bs|-1}& \boldsymbol{\nu}_{\bs ,|\bs|} \end{array}
\right|, \\ \label{eq:ht-D}
\hat{\tilde{D}}_{\bs} & = \left|\begin{array}{cccccc} \boldsymbol{\nu}_{\bs ,0} & \boldsymbol{\nu}_{\bs ,1} & \dots& \boldsymbol{\nu}_{\bs ,|\bs|-3}& \boldsymbol{\nu}_{\bs ,|\bs|-2}& \boldsymbol{\nu}_{\bs ,|\bs|+1} \end{array}
\right|.
\end{align} 
Notice that equation \eqref{eq:mo2-p-1} can be obtained by collecting the terms at $x^{|\bs|-1}$ in identity \eqref{eq:xZD}.

An appropriate Sylvester's identity for the
determinant
\begin{equation*} 
\left| \begin{array}{cccccc} \boldsymbol{\nu}_{\bs+\be_j,0} & \dots &  \boldsymbol{\nu}_{\bs+\be_j,|\bs| - 1} & \boldsymbol{\nu}_{\bs+\be_j,|\bs|}& \boldsymbol{\nu}_{\bs+\be_j ,|\bs|+1} \\
0 & \dots & 0 & 0 & 1 
\end{array}\right|, 
\end{equation*}
gives
\begin{equation} \label{eq:mo2-p-2}
D_{\bs+\be_j} \hat{\tilde{D}}_{\bs} = \tilde{D}_{\bs+\be_j} \tilde{D}_{\bs} - D_{\bs} \hat{D}_{\bs+\be_j}.
\end{equation}
Equivalently, equation \eqref{eq:mo2-p-2} can be obtained by collecting the terms at $x^{|\bs|-1}$ in identity \eqref{eq:ZD-tt}.
Both equations \eqref{eq:mo2-p-1} and \eqref{eq:mo2-p-2}, and definition \eqref{eq:a-b-r} imply that
\begin{equation} \label{eq:a-D-hat}
\sum_{i=1}^r \frac{D_{\bs + \be_i} D_{\bs - \be_i}}{D_{\bs}^2} =  \frac{\tilde{D}_{\bs}}{D_{\bs}}  \left( \frac{\tilde{D}_{\bs + \be_j}}{D_{\bs + \be_j}} -  \frac{\tilde{D}_{\bs}}{D_{\bs}}  \right) - \frac{\hat{D}_{\bs + \be_j}}{D_{\bs + \be_j}} + \frac{\hat{D}_{\bs}}{D_{\bs}} , \qquad j=1,\dots , r.
\end{equation}
To conclude determinantal verification of equation \eqref{eq:mo-2} it is enough to use formula \eqref{eq:Djk-tD} and analogous equation
\begin{equation} \label{eq:D-D-hat}
D_{\bs+\be_j + \be_k} \tilde{D}_{\bs} = \hat{D}_{\bs + \be_k} D_{\bs + \be_j} - \hat{D}_{\bs + \be_j} D_{\bs + \be_k},
\end{equation}
which can be easily derived by appropriate Sylvester identity.
\end{Ex}
\begin{Rem}
	Equation \eqref{eq:D-D-hat} can be obtained by collecting the terms at $x^{|\bs|-1}$ in identity \eqref{eq:Z-D-jk}.
\end{Rem}

In analogy to definitions \eqref{eq:h-D} and \eqref{eq:ht-D} of the functions $\hat{D}_{\bs}$ and $\hat{\tilde{D}}_{\bs}$, and the definition \eqref{eq:tZ} of the polynomial $\tilde{Z}_{\bs}(x)$, let us consider the polynomials  
\begin{align} \label{eq:h-Z}
\hat{Z}_{\bs}(x) & = 	\left| \begin{array}{ccccccc} \boldsymbol{\nu}_{\bs,0} &  \boldsymbol{\nu}_{\bs,1} & \dots & \boldsymbol{\nu}_{\bs,|\bs|-2}& \boldsymbol{\nu}_{\bs,|\bs|}& \boldsymbol{\nu}_{\bs,|\bs|+1} \\
1 & x & \dots & x^{|\bs|-2}& x^{|\bs|} & x^{|\bs|+1} 
\end{array}\right|, \\
\hat{\tilde{Z}}_{\bs}(x) & = 	\left| \begin{array}{ccccccc} \boldsymbol{\nu}_{\bs,0} &  \boldsymbol{\nu}_{\bs,1} & \dots & \boldsymbol{\nu}_{\bs,|\bs|-2}& \boldsymbol{\nu}_{\bs,|\bs|-1}& \boldsymbol{\nu}_{\bs,|\bs|+2} \\
1 & x & \dots & x^{|\bs|-2}& x^{|\bs|-1} & x^{|\bs|+2} 
\end{array}\right|. 
\end{align} 
Then application of the method used in proof of Proposition~\ref{prop-r-term} to the following determinant of order $2|\bs|+2$
\begin{equation} \label{eq:det-Z-Z-ht}
\left| \begin{array}{ccccccccc}
\boldsymbol{\nu}_{\bs,0} &  \boldsymbol{\nu}_{\bs,1} & \dots &  \boldsymbol{\nu}_{\bs, |\bs|} &  \boldsymbol{\nu}_{\bs,|\bs| + 1}   &&& &
\\ 1 & x & \dots & x^{|\bs|} & x^{|\bs|+1} & & & \! \! \boldsymbol{0}^{|\bs|+1}_{|\bs|}  & 
\\ \boldsymbol{\nu}_{\bs,1} &  \boldsymbol{\nu}_{\bs,2} & \dots &  \boldsymbol{\nu}_{\bs, |\bs| +1 } &  \boldsymbol{\nu}_{\bs, |\bs| + 2 } & \boldsymbol{\nu}_{\bs,0} & \boldsymbol{\nu}_{\bs,1} & \dots &  \boldsymbol{\nu}_{\bs, |\bs| -1} \\
x & x^2 & \dots & x^{|\bs|+1} & x^{|\bs|+2} & 1 & x &\dots & x^{|\bs|-1}
\end{array}
\right| ,
\end{equation}
gives the following quadratic identity
	\begin{equation}
\label{eq:MDTE-Z-ht}
Z_{\bs}(x) \hat{\tilde{Z}}_{\bs}(x) - \tilde{Z}_{\bs}(x)^2 + Z_{\bs}(x) \hat{Z}_{\bs}(x) = \sum_{i=1}^r 	Z_{\bs+\be_i}(x) 	Z_{\bs-\be_i}(x).
\end{equation}
\begin{Rem}
Equation \eqref{eq:mo2-p-1} can be obtained by collecting the terms at $x^{2|\bs|}$ in the above equation~\eqref{eq:MDTE-Z-ht}.
\end{Rem}
\begin{Com}
We will present another proof of equations~\eqref{eq:mo-1}-\eqref{eq:mo-3} in more general context in Section~\ref{sec:MTS}, where we introduce convenient notation. We also present there more elegant analogs of functions  given above and of the corresponding identities.
\end{Com}

\section{Discrete-time evolution of the measures}
\label{sec:MTS}
In this Section we go out of the initial space of discrete variables by adding suitable evolution variable, what in the simplest case $r=1$ gives the discrete-time Toda equations~\cite{Hirota-2dT}. In the analogous situation of the Pad\'{e} approximation theory, this additional variable allows us to go to the full Pad\'{e} table from its main diagonal~\cite{Doliwa-HP-MP-T}. The resulting formulas are more symmetric and transparent then those introduced in the previous Section, and can be used to study the original multiple orthogonal polynomials.  
\subsection{Basic formulas}
Let us assume that the measures evolve in discrete time variable $t\in \mathbb{N}_0$ according~\cite{AptekarevDerevyaginMikiVanAssche}  to  equation
\begin{equation}
d\mu_{t}^{(j)}(x) = x^t d\mu^{(j)}(x), \qquad t\in \mathbb{N}_0.
\end{equation} 
Their moments read
\begin{equation}
\nu^{(j)}_{k,t} = \int_\RR x^k d\mu^{(j)}_t(x) = \int_\RR x^{k+j} d\mu^{(j)}(x) = \nu^{(j)}_{k+t},
\end{equation}
and allow to find the corresponding canonical multiple orthogonal polynomials
\begin{equation}
Z_{\bs ,t}(x) = \left| \begin{array}{ccccc} \boldsymbol{\nu}_{\bs,t} &  \boldsymbol{\nu}_{\bs,t+1} & \dots &  \boldsymbol{\nu}_{\bs,t + |\bs| - 1} &  \boldsymbol{\nu}_{\bs,t+|\bs|} \\
1 & x & \dots & x^{|\bs|-1} & x^{|\bs|}
\end{array}
\right|,
\end{equation}
and the determinants
\begin{equation}
D_{\bs , t} = \left|  \boldsymbol{\nu}_{\bs,t} , \boldsymbol{\nu}_{\bs,t+1} ,\dots ,  \boldsymbol{\nu}_{\bs,t + |\bs| - 1}
\right|.
\end{equation}

\begin{Com}
The continuous-time evolution of the measures 
\begin{equation}
\label{eq:mu-t-r}
d\mu^{(j)}(t,x) = e^{-xt} d\mu^{(j)}(x), \qquad t \in \RR_+, \qquad j=1,\dots ,r,
\end{equation}
well known in the theory of continuous-time birth and death processes~\cite{LedermannReuter,KarlinMcGregor}, leads~\cite{Moerbeke,Moser,Flaschka}, in the simplest case of $r=1$, to the standard Toda lattice equation~\cite{Toda-TL}. See \cite{AptekarevDerevyaginMikiVanAssche} for presentation of the generalization to the multiple orthogonal $r>1$ case. 
\end{Com}

\begin{Com}
	In \cite{AptekarevDerevyaginMikiVanAssche} the evolution of the measures was presented in more general context of Christoffel and Geronimus transformations of multiple orthogonal polynomials. Instead of multiplication by $x$ at a single evolution step, one can apply multiplication by a factor $(x-\lambda_t)$, where the parameter $\lambda_t$ may depend on $t$. This corresponds to generation of discrete integrable systems by B\"{a}cklund--Darboux transformations~\cite{LeBen,WahlquistEstabrook}. Because our main object of interest is the original system of multiple orthogonal polynomials, we do not need such a generalization. Notice however, that in the context of Hermite--Pad\'{e} approximation problem such a generalization corresponds to transition from approximation to the interpolation; see relevant discussion in~\cite{Doliwa-NAMTS} supplemented by presentation of the corresponding integrable equations together with their particular solutions. We also remark that Mahler in his fundamental paper~\cite{Mahler-P} considered such more general problem.
\end{Com}

For fixed $t$ the new functions $Z_{\bs ,t}(x) $ and $D_{\bs, t}$ satisfy equations \eqref{eq:Z-D-jk}, \eqref{eq:D-ijk-H} and \eqref{eq:Z-ijk-H}. Similarly one can define $Q_{\bs,t}(x)$ and $C^{(ij)}_{\bs ,t}$, and  correspondingly extend equations~\eqref{eq:Q-C-jk}. In fact, all functions and identities considered in previous Sections have their fixed time analogs. We will be however interested in relations between the functions in neighboring moments of the discrete time $t$. These will provide more insight into the multiple orthogonal polynomials and their relation to integrable systems theory.
\begin{Prop}
	When $r\geq 1$ then the canonical multiple polynomials $Z_{\bs,t}(x)$ and the determinants $D_{\bs,t}$ satisfy the following system analogous to \eqref{eq:Z-D-jk}
\begin{equation} \label{eq:Z-D-j-t}
Z_{\bs+\be_j,t}(x) D_{\bs, t+1} = x Z_{\bs, t+1}(x) D_{\bs + \be_j, t} -  Z_{\bs ,t}(x) D_{\bs+\be_j, t+1} , \qquad j=1,\dots r.
\end{equation}
\end{Prop}
\begin{proof}
	By application of the Sylvester identity to $Z_{\bs+\be_j,t}(x)$, its first and last columns, the last row of $j$th block and the last row of the whole matrix.
\end{proof}
\begin{Cor}
	When the determinants $D_{\bs,t}$, $D_{\bs + \be_j,t}$ and $D_{\bs, t+1 }$ do not vanish then equation~\eqref{eq:Z-D-j-t} gives the corresponding analog of equations~\eqref{eq:Q-C-jk}
	\begin{equation} \label{eq:Q-A-j-t}
	xQ_{\bs,t+1}(x)  = Q_{\bs + \be_j,t}(x) + A^{(j)}_{\bs,t} Q_{\bs,t}(x), \qquad j=1,\dots r,	
	\end{equation}
	where 
	\begin{equation} \label{eq:A-Q-r}
	A^{(j)}_{\bs,t} = \frac{D_{\bs + \be_j,t+1}D_{\bs,t}}{D_{\bs+\be_j , t} D_{\bs ,t+1}}.
	\end{equation}
\end{Cor}
\begin{Rem}
	Equations \eqref{eq:Q-A-j-t} were obtained in  \cite{AptekarevDerevyaginMikiVanAssche}, using analytic arguments based on the orthogonality relations.
\end{Rem}
\begin{Rem}
	The discrete-time variable $t$ plays here distinguished role. The symmetry with the discrete variables $s_j$ can be restored~\cite{Doliwa-Siemaszko-2} in the context of the Hermite--Pad\'{e} approximation interpretation of the multiple orthogonal polynomials.
\end{Rem}
By application of the Sylvester identity to the determinant $D_{\bs+\be_j+\be_k,t}$, its first and last columns, and the last rows of $j$th block and $k$th blocks we obtain the analog of \eqref{eq:D-ijk-H}
\begin{Prop} \label{prop:DD-t}
	When $r\geq 2$ then the determinants $D_{\bs,t}$ satisfy the following system
	\begin{equation} \label{eq:D-jk-t-H}
	D_{\bs+\be_j + \be_k,t} 	D_{\bs,t+1} = 	D_{\bs+ \be_k,t+1} 	D_{\bs+\be_j,t } - 	D_{\bs+\be_j, t+1} 	D_{\bs+\be_k, t}, \qquad j<k.	
	\end{equation}
\end{Prop}
\begin{Com}
	Equation \eqref{eq:D-jk-t-H} can be also obtained by combination of two copies of equations \eqref{eq:Z-D-j-t} for indices $j$ and $k$, and of equation \eqref{eq:Z-D-jk}, provided $Z_{\bs,t}(x)$ does not vanish.
\end{Com}
The following results, obvious from the multiple orthogonal polynomials point of view, will be relevant in  the more general context of integrable systems, which we consider in Section~\ref{sec:MOP-I}. They state that the discrete-time evolution given by \eqref{eq:D-jk-t-H} is a symmetry of the linear problem~\eqref{eq:Z-D-jk} and of the Hirota equation~\eqref{eq:D-ijk-H} itself. 

\begin{Com}
	By excluding polynomials in the discrete-time $t+1$ in three copies of equation~\eqref{eq:D-jk-t-H} for pairs of indices $i<j<k$, we obtain the Hirota equation~\eqref{eq:D-ijk-H} for a fixed value of $t$. Analogously one can obtain the same equation, but for time $t+1$.
\end{Com}
\begin{Com}
The linear problem~\eqref{eq:Z-D-jk} in time $t$, and evolution equation~\eqref{eq:Z-D-j-t} imply validity of the linear problem in time $t+1$. 	
\end{Com}

By collecting coefficients at $x^{|\bs|}$ in equation~\eqref{eq:Z-D-j-t} we obtain
\begin{equation} \label{eq:D-tD-t}
\tilde{D}_{\bs+\be_j,t} D_{\bs, t+1} = \tilde{D}_{\bs,t+1} D_{\bs+\be_j, t} + D_{\bs,t} D_{\bs+\be_j, t+1}, \qquad j=1,\dots , r.
\end{equation}
Notice that equation~\eqref{eq:D-tD-t} can be also obtained by application of the Sylvester identity to matrix $M_{\bs+\be_j,t}$ supplemented from below by the row $(0,\dots , 0, 1, 0)$, with the first and the last columns and the last row of $j$th row-block and the last row of the whole matrix. 

\begin{Com}
	Equations \eqref{eq:D-tD-t} and \eqref{eq:D-jk-t-H} imply equations~\eqref{eq:Djk-tD} both in time $t$ as in time $t+1$, i.e. the discrete-time evolution of the functions $\tilde{D}_{\bs,t}$ is consistent with their defining equations~\eqref{eq:Djk-tD}.	
\end{Com}

The following non-linear relation of Hirota's type seems to be new in the context of the multiple orthogonal polynomials.
\begin{Prop}
	When $r\geq 2$ then the polynomials $Z_{\bs,t}(x)$ satisfy the following system
	\begin{equation} \label{eq:Z-jk-t-H}
	Z_{\bs+\be_j + \be_k,t}(x) 	Z_{\bs,t+1}(x) = 	Z_{\bs+ \be_k,t+1}(x) 	Z_{\bs+\be_j,t }(x) - 	Z_{\bs+\be_j, t+1}(x) 	Z_{\bs+\be_k, t}(x), \qquad j<k.	
	\end{equation}
\end{Prop}
\begin{proof}
	Apply the Sylvester identity to $Z_{\bs+\be_j+\be_k,t}(x)$, its first and last columns, and the last rows of $j$th block and $k$th blocks. 
\end{proof}

\begin{Rem}
	Equation \eqref{eq:D-jk-t-H} can be derived from \eqref{eq:Z-jk-t-H} by collecting the highest order term coefficients. 
\end{Rem}
\begin{Com}
	Equation \eqref{eq:Z-jk-t-H} can be also obtained by combination of two copies of (shifted) equations \eqref{eq:Z-D-j-t} for indices $j$ and $k$, and of equation \eqref{eq:Z-D-jk}, provided $D_{\bs,t}(x)$ does not vanish. 
	By excluding shifts in the discrete-time $t$ in three copies of equation~\eqref{eq:Z-jk-t-H} for pairs from the three indices $i<j<k$, we obtain equation~\eqref{eq:Z-ijk-H} for a fixed value of $t$.
\end{Com}

\subsection{Discrete-time Toda equations as the Paszkowski-type relations for multiple orthogonal polynomials}
The present Section is devoted to more complicated equations satisfied by the multiple orthogonal polynomials, which are similar to the multi-term linear equations~\eqref{eq:lin-term-r} but involve also the discrete-time evolution direction. From the point of view of the Hermite--Pad\'{e} approximation problem they correspond to the Paszkowski constraint~\cite{Paszkowski,Baker} and related equations given in~\cite{Doliwa-Siemaszko-2}.

\begin{Prop}
	When $r\geq 1$ then the canonical multiple polynomials $Z_{\bs,t}(x)$ and the determinants $D_{\bs,t}$ satisfy the following system
\begin{equation} \label{eq:Z-D-P-j-t}
Z_{\bs,t+1}(x) D_{\bs, t} = Z_{\bs, t}(x) D_{\bs, t+1} -  \sum_{i=1}^r Z_{\bs - \be_i,t+1}(x) D_{\bs+\be_i, t} .
\end{equation}
\end{Prop}
\begin{proof}
Apply the method used in proof of Proposition~\ref{prop-r-term} to the following determinant of order $2|\bs|+1$
\begin{equation} \label{eq:det-Z-D}
\left| \begin{array}{ccccccccc}
\boldsymbol{\nu}_{\bs,t} &  \boldsymbol{\nu}_{\bs,t+1} & \dots &  \boldsymbol{\nu}_{\bs,t + |\bs|-1 } &  \boldsymbol{\nu}_{\bs,t + |\bs| } & & & \boldsymbol{0}^{|\bs|}_{|\bs|}  & \\\boldsymbol{\nu}_{\bs,t+1} &  \boldsymbol{\nu}_{\bs,t+2} & \dots &  \boldsymbol{\nu}_{\bs,t + |\bs| } &  \boldsymbol{\nu}_{\bs,t + |\bs| + 1} &  \boldsymbol{\nu}_{\bs,t+1} &  \boldsymbol{\nu}_{\bs,t+2} & \dots &  \boldsymbol{\nu}_{\bs,t + |\bs| }  \\
1 & x &  \dots & x^{|\bs|-1} & x^{|\bs|} & 1 & x & \dots & x^{|\bs|-1}
\end{array}
\right| .
\end{equation}
\end{proof}
\begin{Cor}
	When the determinants $D_{\bs,t}$, $D_{\bs,t+1}$ and $D_{\bs - \be_j, t+1 }$, $j=1,\dots ,r$, do not vanish then equation~\eqref{eq:Z-D-P-j-t} leads directly to the one, obtained in \cite{AptekarevDerevyaginMikiVanAssche}	by another technique,
	\begin{equation} \label{eq:Q-B-j-t}
	Q_{\bs,t+1}(x)  = Q_{\bs,t}(x) - \sum_{i=1}^r B^{(i)}_{\bs,t} Q_{\bs - \be_i,t+1}(x),
	\end{equation}
	where 
\begin{equation} \label{eq:B-Q-r}
B^{(j)}_{\bs,t} = \frac{D_{\bs + \be_j,t}D_{\bs - \be_j,t+1}}{D_{\bs, t} D_{\bs ,t+1}}.
\end{equation}
\end{Cor}

The following relation can be obtained by collecting coefficients at $x^{|\bs|-1}$ in equation~\eqref{eq:Z-D-P-j-t}, but in the spirit of our paper we give its determinantal proof.
\begin{Prop} \label{prop:D-tD-t}
	The discrete-time evolution of the determinants is given by
\begin{equation} \label{eq:tD-D-t}
\tilde{D}_{\bs,t+1} D_{\bs,t} - D_{\bs,t+1} \tilde{D}_{\bs,t} = \sum_{i=1}^r D_{\bs + \be_i,t} D_{\bs-\be_i,t+1} .
\end{equation}		
\end{Prop}
\begin{proof}
Apply the method used in proof of Proposition~\ref{prop-r-term} to the following determinant of order $2|\bs|$
\begin{equation*} \label{eq:det-D-D-2}
\left| \begin{array}{ccccccccc}
\boldsymbol{\nu}_{\bs,t} &  \boldsymbol{\nu}_{\bs,t+1} & \dots &  \boldsymbol{\nu}_{\bs,t + |\bs|-1 } &  \boldsymbol{\nu}_{\bs,t + |\bs|} & & & \boldsymbol{0}^{|\bs|}_{|\bs|-1}  &  \\ \boldsymbol{\nu}_{\bs,t+1} &  \boldsymbol{\nu}_{\bs,t+2} & \dots &  \boldsymbol{\nu}_{\bs,t + |\bs|  } &  \boldsymbol{\nu}_{\bs,t + |\bs| +1 } &  \boldsymbol{\nu}_{\bs,t+1} &  \boldsymbol{\nu}_{\bs,t+2} & \dots &  \boldsymbol{\nu}_{\bs,t + |\bs| -1}
\end{array}
\right| .
\end{equation*}
\end{proof}
By the duality between the determinants $Z_{\bs,t}(x)$ and $D_{\bs,t}$  we may expect an analogous quadratic equation involving the polynomials only. 
\begin{Prop}
	When $r\geq 1$ then the canonical multiple polynomials $Z_{\bs,t}(x)$ and the polynomials $\tilde{Z}_{\bs , t}(x)$ satisfy the following equations
	\begin{equation}
	\label{eq:MDTE-ZZ-t}
	Z_{\bs,t}(x) \tilde{Z}_{\bs , t+1}(x)  -	\tilde{Z}_{\bs,t}(x) 	Z_{\bs,t+1 }(x) = \sum_{i=1}^r 	Z_{\bs+\be_i, t}(x) 	Z_{\bs-\be_i, t+1}(x).
	\end{equation}	
\end{Prop}
\begin{proof}
	Apply the method used in proof of Proposition~\ref{prop-r-term} to the following determinant of order $2|\bs|+2$
	\begin{equation} \label{eq:det-Z-Z-t}
	\left| \begin{array}{ccccccccc}
	\boldsymbol{\nu}_{\bs,t} &  \boldsymbol{\nu}_{\bs,t+1} & \dots &  \boldsymbol{\nu}_{\bs,t + |\bs|} &  \boldsymbol{\nu}_{\bs,t + |\bs|+1}   &&& &
	\\ 1 & x & \dots & x^{|\bs|} & x^{|\bs|+1} & & & \! \! \boldsymbol{0}^{|\bs|+1}_{|\bs|}  & 
	\\ \boldsymbol{\nu}_{\bs,t+1} &  \boldsymbol{\nu}_{\bs,t+2} & \dots &  \boldsymbol{\nu}_{\bs,t + |\bs| +1} &  \boldsymbol{\nu}_{\bs,t + |\bs| +2 } & \boldsymbol{\nu}_{\bs,t+1} & \boldsymbol{\nu}_{\bs,t+2} & \dots &  \boldsymbol{\nu}_{\bs,t + |\bs|} \\
	1 & x & \dots & x^{|\bs|} & x^{|\bs|+1} & 1 & x & \dots & x^{|\bs|-1}
	\end{array}
	\right| .
	\end{equation}
\end{proof}
\begin{Rem}
	The identity \eqref{eq:tD-D-t} can be obtained from \eqref{eq:MDTE-ZZ-t} by collection of its coefficients at $x^{2|\bs|}$.
\end{Rem}

Equation \eqref{eq:Z-D-P-j-t} involves the polynomials in two time steps. It can be considered as analogous of the multi-term constraint~\eqref{eq:3-term-r-Z} or \eqref{eq:xZD} at fixed time. Equation~\eqref{eq:tD-D-t} contains two types of determinants at two time steps. 
The following bilinear equation, stated in~\cite{AptekarevDerevyaginMikiVanAssche}, involves one type of determinants but in three different time steps. 
\begin{Prop}
	When $r\geq 1$ and $t\geq 1$ then the  determinants $D_{\bs,t}$ satisfy equations
\begin{equation} \label{eq:MDTE-D}
D_{\bs,t}^2 = 	D_{\bs,t+1} 	D_{\bs,t-1 } - \sum_{i=1}^r 	D_{\bs+\be_i, t-1} 	D_{\bs-\be_i, t+1},
\end{equation}	
\end{Prop}
\begin{proof}
Apply the method used in proof of Proposition~\ref{prop-r-term} to the following determinant of order $2|\bs|$
\begin{equation} \label{eq:det-D-D}
\left| \begin{array}{ccccccccc}
\boldsymbol{\nu}_{\bs,t-1} &  \boldsymbol{\nu}_{\bs,t} & \dots &  \boldsymbol{\nu}_{\bs,t + |\bs|-2 } &  \boldsymbol{\nu}_{\bs,t + |\bs| -1 } & & & \boldsymbol{0}^{|\bs|}_{|\bs|-1}  &  \\ \boldsymbol{\nu}_{\bs,t} &  \boldsymbol{\nu}_{\bs,t+1} & \dots &  \boldsymbol{\nu}_{\bs,t + |\bs| -1 } &  \boldsymbol{\nu}_{\bs,t + |\bs| } &  \boldsymbol{\nu}_{\bs,t+1} &  \boldsymbol{\nu}_{\bs,t+2} & \dots &  \boldsymbol{\nu}_{\bs,t + |\bs| -1}
\end{array}
\right| .
\end{equation}
\end{proof}
\begin{Rem}
	In the simplest case $r=1$ the above equation is known as the discrete-time Toda lattice equation in Hirota's bilinear form~\cite{Hirota-2dT}.
\end{Rem}

The symmetry between the determinants $Z_{\bs,t}(x)$ and $D_{\bs,t}$ suggests an analogous equation on the level of multiple orthogonal polynomials, which was not considered earlier.
\begin{Prop}
	When $r\geq 1$ and $t\geq 1$ then the canonical multiple polynomials $Z_{\bs,t}(x)$ and the determinants $D_{\bs,t}$ satisfy equations
	\begin{equation}
	\label{eq:MDTE-Z}
	Z_{\bs,t}(x)^2  = 	Z_{\bs,t+1}(x) 	Z_{\bs,t-1 }(x) - \sum_{i=1}^r 	Z_{\bs+\be_i, t-1}(x) 	Z_{\bs-\be_i, t+1}(x).
	\end{equation}	
\end{Prop}
\begin{proof}
	Apply the method used in proof of Proposition~\ref{prop-r-term} to the following determinant of order $2|\bs|+2$
\begin{equation} \label{eq:det-Z-Z}
\left| \begin{array}{ccccccccc}
\boldsymbol{\nu}_{\bs,t-1} &  \boldsymbol{\nu}_{\bs,t} & \dots &  \boldsymbol{\nu}_{\bs,t + |\bs|-1} &  \boldsymbol{\nu}_{\bs,t + |\bs|}   &&&& 
\\ 1 & x & \dots & x^{|\bs|} & x^{|\bs|+1} & & & \! \! \boldsymbol{0}^{|\bs|+1}_{|\bs|}  & 
\\ \boldsymbol{\nu}_{\bs,t} &  \boldsymbol{\nu}_{\bs,t+1} & \dots &  \boldsymbol{\nu}_{\bs,t + |\bs| } &  \boldsymbol{\nu}_{\bs,t + |\bs| +1 } & \boldsymbol{\nu}_{\bs,t+1} & \boldsymbol{\nu}_{\bs,t+2} & \dots &  \boldsymbol{\nu}_{\bs,t + |\bs|} \\
1 & x & \dots & x^{|\bs|} & x^{|\bs|+1} & x & x^2 & \dots & x^{|\bs|}
\end{array}
\right| .
\end{equation}
\end{proof}
\begin{Rem}
	The identity \eqref{eq:MDTE-D} can be obtained from \eqref{eq:MDTE-Z} by collection of its highest order terms.
\end{Rem}
In order to complete the presentation let us give the corresponding linear equation, which also seems not to be known before.

\begin{Prop}
	When $r\geq 1$ and $t\geq 1$ then the canonical multiple polynomials $Z_{\bs,t}(x)$ and the determinants $D_{\bs,t}$ satisfy the following system
	\begin{equation} \label{eq:Z-D-P-j-t2}
	Z_{\bs,t}(x) D_{\bs, t} = Z_{\bs, t-1}(x) D_{\bs, t+1} -  x \sum_{i=1}^r Z_{\bs - \be_i,t+1}(x) D_{\bs+\be_i, t-1} .
	\end{equation}
\end{Prop}
\begin{proof}
	Apply the method used in proof of Proposition~\ref{prop-r-term} to the following determinant of order $2|\bs|+1$
	\begin{equation} \label{eq:det-Z-D-2}
	\left| \begin{array}{ccccccccc}
	\boldsymbol{\nu}_{\bs,t-1} &  \boldsymbol{\nu}_{\bs,t} & \dots &  \boldsymbol{\nu}_{\bs,t + |\bs|-2 } &  \boldsymbol{\nu}_{\bs,t + |\bs| -1 } & & & \! \! \boldsymbol{0}^{|\bs|}_{|\bs|}  & \\\boldsymbol{\nu}_{\bs,t} &  \boldsymbol{\nu}_{\bs,t+1} & \dots &  \boldsymbol{\nu}_{\bs,t + |\bs| -1 } &  \boldsymbol{\nu}_{\bs,t + |\bs| } &  \boldsymbol{\nu}_{\bs,t+1} &  \boldsymbol{\nu}_{\bs,t+2} & \dots &  \boldsymbol{\nu}_{\bs,t + |\bs| }  \\
	1 & x &  \dots & x^{|\bs|-1} & x^{|\bs|} & x & x^2 & \dots & x^{|\bs|}
	\end{array}
	\right| .
	\end{equation}
\end{proof}
\begin{Rem}	
	Equation \eqref{eq:MDTE-D} can be obtained from \eqref{eq:Z-D-P-j-t2} by collecting its highest order terms. 
\end{Rem}

\subsection{The compatibility conditions for the multi-term relation}
By extending the multi-term relation~\eqref{eq:lin-term-r} to allowed values of the discrete-time parameter $t$ we have
\begin{equation} \label{eq:3-term-r-t}
x Q_{\bs,t} (x) = Q_{\bs + \be_j,t}(x) + b_{\bs,t}^{(j)} Q_{\bs,t}(x) + \sum_{k=1}^r a_{\bs,t}^{(k)} Q_{\bs - \be_k,t}(x), \qquad j=1,2,\dots , r,
\end{equation}
where the recurrence coefficients are given there by formulas of the form~\eqref{eq:a-b-r} 
	\begin{equation} \label{eq:a-b-r-t} 
	a_{\bs,t}^{(j)} = \frac{D_{\bs + \be_j,t} D_{\bs - \be_j,t}}{D_{\bs,t}^2}, \qquad 
	b_{\bs,t}^{(j)} = \frac{\tilde{D}_{\bs + \be_j,t}}{D_{\bs + \be_j,t}} - \frac{\tilde{D}_{\bs,t}}{D_{\bs,t}}.
	\end{equation}
Similarly as it was done in~\cite{AptekarevDerevyaginMikiVanAssche}, by excluding the polynomials in time $(t+1)$ from equations~\eqref{eq:Z-D-P-j-t} using relation~\eqref{eq:Z-D-j-t} one can the following expressions of the coefficients such that $a_{\bs,t}^{(j)}$ agrees with the above, but $b_{\bs,t}^{(j)}$ is given by
	\begin{equation} \label{eq:b-D-D-r} 
	b_{\bs,t}^{(j)} = \frac{D_{\bs + \be_j,t+1} D_{\bs,t}}{D_{\bs+\be_j,t}D_{\bs,t+1}} +
	\sum_{i=1}^r \frac{D_{\bs + \be_i,t} D_{\bs-\be_i,t+1}} {D_{\bs,t}D_{\bs,t+1}}.
	\end{equation}
\begin{Rem}
Both expressions \eqref{eq:a-b-r-t} and \eqref{eq:b-D-D-r} for $b_{\bs,t}^{(j)}$ are equal by construction. In the spirit of the present paper, by using equation \eqref{eq:D-tD-t} let us rewrite \eqref{eq:tD-D-t} in the form
	\begin{equation*}
	\frac{\tilde{D}_{\bs + \be_j , t}}{D_{\bs + \be_j, t}} -\frac{\tilde{D}_{\bs , t}}{D_{\bs, t}} = \frac{D_{\bs + \be_j,t+1} D_{\bs,t}}{D_{\bs+\be_j,t}D_{\bs,t+1}} +
	\sum_{i=1}^r \frac{D_{\bs + \be_i,t} D_{\bs-\be_i,t+1}} {D_{\bs,t}D_{\bs,t+1}},
	\end{equation*}
what by Proposition~\ref{prop:D-tD-t} provides determinantal proof of the relation.
\end{Rem}

	Finally, as we promised at the end of Section~\ref{sec:G3T}, let us consider extended version of equations~\eqref{eq:mo-1}-\eqref{eq:mo-3} in discrete-time $t$
	\begin{align} \label{eq:mo-1-t}
	\quad \; b_{\bs+\be_k,t}^{(j)} - b_{\bs + \be_j,t}^{(k)} & = b_{\bs,t}^{(j)} - b_{\bs,t}^{(k)},\\  \label{eq:mo-2-t}
	b_{\bs,t}^{(k)} b_{\bs + \be_k,t}^{(j)} - b_{\bs,t}^{(j)} b_{\bs + \be_j,t}^{(k)}  & = \sum_{i=1}^r \left( a_{\bs+\be_k,t}^{(i)} - a_{\bs+\be_j,t}^{(i)}\right),\\  \label{eq:mo-3-t}
	a_{\bs + \be_k,t}^{(j)}( b_{\bs -\be_j,t}^{(j)} - b_{\bs - \be_j,t}^{(k)}) & = a_{\bs,t}^{(j)}(b_{\bs,t}^{(j)} - b_{\bs,t}^{(k)}).
	\end{align}
Equation~\eqref{eq:mo-2-t}, which was the most difficult to demonstrate in the previous formalism, can be verified with the help of \eqref{eq:b-D-D-r}, where its proof reduces to application of the bilinear Hirota equations in the forms \eqref{eq:D-ijk-H} and \eqref{eq:D-jk-t-H}. Technically, the polynomials $Z_{\bs,t+1}(x)$ and functions $D_{\bs,t+1}$ can be considered as counterparts of $\tilde{Z}_{\bs,t}(x)$ and $\tilde{D}_{\bs,t}$, because both of them are obtained by removing certain columns in corresponding determinants. One can extend this analogy to the polynomials $\hat{Z}_{\bs,t}(x)$ and $\hat{\tilde{Z}}_{\bs,t}(x)$, and to the functions  $\hat{D}_{\bs,t}$ and $\hat{\tilde{D}}_{\bs,t}$ defined in Section~\ref{sec:G3T} for $t=0$. However embedding the identities discussed in there into more general system with discrete-time $t$ clarifies the picture. 

\section{Integrability of the discrete-time multidimensional Toda system} \label{sec:MOP-I}
In the present Section we do not put any restrictions of the range of discrete-space variables $s_j$, $j=1,\dots ,r$, and of the discrete-time variable $t$, i.e. they are allowed to be arbitrary integer numbers. We also abandon determinantal interpretation of the functions involved.

\subsection{The multifield form of the equations}

Let us start from the linear problem composed out of equations~\eqref{eq:Q-A-j-t} and \eqref{eq:Q-B-j-t}, which we repeat for convenience
	\begin{align} \tag{\ref{eq:Q-A-j-t}}
xQ_{\bs,t+1}(x)&  = Q_{\bs + \be_j,t}(x) + A^{(j)}_{\bs,t} Q_{\bs,t}(x), \qquad j=1,\dots r,\\ \tag{\ref{eq:Q-B-j-t}}
	Q_{\bs,t+1}(x) & = Q_{\bs,t}(x) - \sum_{i=1}^r B^{(i)}_{\bs,t} Q_{\bs - \be_i,t+1}(x),
\end{align}
In the general context of the theory of integrable systems  $Q_{\bs,t}(x)$ is called the wave function, and do not has to be polynomial in $x$, which is called then the spectral parameter. Their precise meaning depends on the class of solutions of the corresponding integrable equations we are interested in. For example, for the so-called finite-gap solutions in terms of Riemann theta functions the parameter $x$ belongs to an algebraic curve~\cite{Krichever-alg-geom-discrete}. 

The functions $A^{(j)}_{\bs,t}$ and $B^{(j)}_{\bs,t} $, $j=1,\dots, r$, need to satisfy equations which enable the wave functions to exist for all values of the spectral parameter. 
Compatibility of the equations \eqref{eq:Q-A-j-t} between themselves implies  \cite{AptekarevDerevyaginMikiVanAssche} the following system
\begin{equation} \label{eq:AA-r}
A_{\bs,t+1}^{(j)} - A_{\bs,t+1}^{(k)} = A_{\bs+\be_k,t}^{(j)} - A_{\bs+\be_j,t}^{(k)}, 
\quad A_{\bs+\be_k,t}^{(j)} A_{\bs,t}^{(k)} = A_{\bs+\be_j,t}^{(k)} A_{\bs,t}^{(j)}.
\end{equation}	
When also equations \eqref{eq:Q-B-j-t} are involved then we obtain in addition~\cite{AptekarevDerevyaginMikiVanAssche}
 the following evolution equations 
	\begin{align}
\label{eq:ABC-r}
A_{\bs-\be_j,t+1}^{(j)} B_{\bs,t+1}^{(j)} & = A_{\bs,t}^{(j)} B_{\bs,t}^{(j)}, \\
\label{eq:BAA-r}
B^{(k)}_{\bs,t}\left( A^{(k)}_{\bs,t} - A^{(j)}_{\bs,t} \right) &  = B^{(k)}_{\bs + \be_j,t}\left( A^{(k)}_{\bs-\be_k,t+1} - A^{(j)}_{\bs-\be_k,t+1} \right),\\
 \label{eq:A+B-r}
A_{\bs,t+1}^{(j)} + \sum_{i=1}^r B_{\bs,t+1}^{(i)} & = A_{\bs,t}^{(j)} + \sum_{i=1}^r B_{\bs+\be_j,t}^{(i)}. 
\end{align}
\begin{Com}
	When making the compatibility condition of the above linear system one arrives \cite{AptekarevDerevyaginMikiVanAssche} to the multi-term relation in the form
\begin{equation} \label{eq:lin-term-AB}
	xQ_{\bs,t}(x) = Q_{\bs+\be_j,t}(x) + \left(A^{(j)}_{\bs,t} +  \sum_{i=1}^r B^{(i)}_{\bs,t} \right) Q_{\bs,t}(x) + \sum_{i=1}^r A^{(i)}_{\bs-\be_i,t} B^{(i)}_{\bs,t} \; Q_{\bs-\be_i,t}(x),
	\end{equation}
what gives the following form of the coefficients $a_{\bs,t}^{(j)}$ and $b_{\bs,t}^{(j)}$
\begin{equation} \label{eq:ab-AB-r} 
a_{\bs,t}^{(j)}=A^{(j)}_{\bs-\be_j,t} B^{(j)}_{\bs,t}, \qquad 
b_{\bs,t}^{(j)} = A^{(j)}_{\bs,t} +  \sum_{i=1}^r B^{(i)}_{\bs,t} .
\end{equation}
Then equations \eqref{eq:ABC-r}-\eqref{eq:A+B-r} imply compatibility conditions~\eqref{eq:mo-1-t}-\eqref{eq:mo-3-t} of the multi-term relations~\eqref{eq:lin-term-r} both in time $t$ and time $(t+1)$, as can be checked by direct verification.
\end{Com}
\begin{Rem}
Notice that the verification of equations~\eqref{eq:mo-1}-\eqref{eq:mo-3}, given in  Section~\ref{sec:G3T}, retains its validity for arbitrary value of $t\in\NN_0$. Introduction~\cite{AptekarevDerevyaginMikiVanAssche} of the discrete-time parameter $t$, what is needed to define the coefficients $A^{(j)}_{\bs,t}$ and $B^{(j)}_{\bs,t}$ by \eqref{eq:A-Q-r} and \eqref{eq:B-Q-r}, allows not only to embed the theory of multiple orthogonal polynomials into the theory of integrable systems but also to perform calculations in a more structured way.
\end{Rem}

\begin{Com}
For $r>1$ the first part~\eqref{eq:Q-A-j-t} of the linear system implies equations \eqref{eq:Q-C-jk} in time $t$
	\begin{equation} \label{eq:Q-C-jk-t}
Q_{\bs+\be_j,t}(x) - Q_{\bs+\be_k,t}(x)= C_{\bs,t}^{(jk)}  Q_{\bs,t}(x) , \qquad j<k,
\end{equation}
where
\begin{equation} \label{eq:AC-r-t}
C_{\bs,t}^{(jk)} = A_{\bs,t}^{(k)} - A_{\bs,t}^{(j)}.
\end{equation}
Compatibility of the system \eqref{eq:Q-C-jk-t} with respect to the variables $s_j$, $j=1,\dots ,r$, leads to one of standard forms \cite{Nimmo-NCKP,KNS-rev} of the discrete Kadomtsev--Petviashvili equations
\begin{equation} \label{eq:CC-r}
C_{\bs,t}^{(ij)} + C_{\bs,t}^{(jk)} = C_{\bs,t}^{(ik)}, 
\quad C_{\bs+\be_k,t}^{(ij)} C_{\bs,t}^{(ik)} = C_{\bs+\be_j,t}^{(ik)} C_{\bs,t}^{(ij)}.
\end{equation}		
In our presentation they are consequence of equations~\eqref{eq:AA-r}. However the deeper reason is that the new variable $t$ should be considered as the parameter of a symmetry of the system~\eqref{eq:CC-r}.
	
\end{Com}

\subsection{The $\tau$-function formalism} Among many identities presented in Section~\ref{sec:MTS} we should select the ones which we consider fundamental. They have to recover all the formulas derived above in the integrable systems context, provided we keep the relations \eqref{eq:A-Q-r} and \eqref{eq:B-Q-r} between $A_{\bs,t}^{(j)}$ and $B_{\bs,t}^{(j)}$, and the $\tau$-function $D_{\bs,t}$.

The $\tau$-function counterpart of the linear problem \eqref{eq:Q-A-j-t}, \eqref{eq:Q-B-j-t} are equations \eqref{eq:Z-D-j-t} and \eqref{eq:Z-D-P-j-t}. It is convenient to consider also equation \eqref{eq:Z-D-P-j-t2} and study mutual relations between these three ingredients
\begin{align}
\tag{\ref{eq:Z-D-j-t}}
Z_{\bs+\be_j,t}(x) D_{\bs, t+1} & = x Z_{\bs, t+1}(x) D_{\bs + \be_j, t} -  Z_{\bs ,t}(x) D_{\bs+\be_j, t+1} , \qquad j=1,\dots r, \\
\tag{\ref{eq:Z-D-P-j-t}}
Z_{\bs,t+1}(x) D_{\bs, t} & = Z_{\bs, t}(x) D_{\bs, t+1} -  \sum_{i=1}^r Z_{\bs - \be_i,t+1}(x) D_{\bs+\be_i, t},\\
\tag{\ref{eq:Z-D-P-j-t2}}
Z_{\bs,t}(x) D_{\bs, t} & = Z_{\bs, t-1}(x) D_{\bs, t+1} -  x \sum_{i=1}^r Z_{\bs - \be_i,t+1}(x) D_{\bs+\be_i, t-1} .
\end{align}

The Hermite--Pad\'{e} analog of the additional linear system~\eqref{eq:Z-D-P-j-t2} was introduced by Paszkowski~\cite{Paszkowski}, and studied in~\cite{Doliwa-Siemaszko-2} on the level of integrable systems. For multiple orthogonal polynomials equation~\eqref{eq:Z-D-P-j-t2} can be derived using determinantal identities, as demonstrated in Section~\ref{sec:MTS}. We will show that it leads to a generalization of the multidimensional discrete-time Toda constraint~\eqref{eq:MDTE-D}. Let us first consider the mutual interrelation of equations~\eqref{eq:Z-D-P-j-t} and \eqref{eq:Z-D-P-j-t2}.
\begin{Prop}
	The system composed of \eqref{eq:Z-D-j-t} and \eqref{eq:Z-D-P-j-t} is equivalent to the system composed of \eqref{eq:Z-D-j-t} and \eqref{eq:Z-D-P-j-t2} provided equation~\eqref{eq:MDTE-D} holds.
\end{Prop}
\begin{proof}
	Shift back equation~\eqref{eq:Z-D-j-t} by $\be_j$, and equation~\eqref{eq:Z-D-P-j-t} one step back in time $t$. Then suitable linear combination, which cancels the sum in \eqref{eq:Z-D-P-j-t} and produces new ones, gives 
	\begin{equation*}
	Z_{\bs,t}(x) \Bigl( 	D_{\bs,t+1} 	D_{\bs,t-1 } - \sum_{i=1}^r 	D_{\bs+\be_i, t-1} 	D_{\bs-\be_i, t+1} \Bigr)  = D_{\bs,t} \Bigl(Z_{\bs, t-1}(x) D_{\bs, t+1} -  x \sum_{i=1}^r Z_{\bs - \be_i,t+1}(x) D_{\bs+\be_i, t-1}\Bigr).
	\end{equation*}	
	The above formula, when compared with equation \eqref{eq:Z-D-P-j-t2}, results in the first part of the statement. The backward implication can be obtained similarly.
\end{proof}

Compatibility of equations~\eqref{eq:Z-D-j-t} between themselves leads, for $r\geq 2$, to the following $\tau$-function version of equations~\eqref{eq:AA-r}
\begin{equation*}
\left(\frac{D_{\bs+ \be_j,t+1} 	D_{\bs+\be_i,t } - 	D_{\bs+\be_i, t+1} 	D_{\bs+\be_j, t}} {
	D_{\bs+\be_j + \be_i,t} 	D_{\bs,t+1}} \right)_{t \to t+1} = \frac{D_{\bs+ \be_j,t+1} 	D_{\bs+\be_i,t } - 	D_{\bs+\be_i, t+1} 	D_{\bs+\be_j, t}} {
	D_{\bs+\be_j + \be_i,t} 	D_{\bs,t+1}} ,
\end{equation*}
where the subscript $t\to t+1$ means that $t$ in the expression in parentheses should be replaced by $t+1$. The above implies the following generalization of equations \eqref{eq:D-jk-t-H}
\begin{equation} \label{eq:D-F-t}
D_{\bs+ \be_j,t+1} 	D_{\bs+\be_i,t } - 	D_{\bs+\be_i, t+1} 	D_{\bs+\be_j, t} = F^{(ij)}_{\bs} D_{\bs+\be_j + \be_i,t} 	D_{\bs,t+1},
\end{equation}
where $F^{(ij)}_{\bs} = - F^{(ji)}_{\bs}$ is a function of the space variables $\bs=(s_1, \dots , s_r)$ only. This, by \eqref{eq:Z-D-j-t} gives in turn the following generalization of the linear problem~\eqref{eq:Z-D-jk}
\begin{equation} \label{eq:Z-D-jk-F}
Z_{\bs+\be_j,t}(x) D_{\bs + \be_k,t} - Z_{\bs+\be_k,t}(x) D_{\bs + \be_j,t}=  F^{(jk)}_{\bs} Z_{\bs,t}(x) D_{\bs+\be_j + \be_k,t},
\end{equation}
and then the corresponding generalized Hirota equations
\begin{equation} \label{eq:D-ijk-F}
F^{(ij)}_{\bs} D_{\bs+\be_i + \be_j,t} D_{\bs + \be_k,t} - F^{(ik)}_{\bs} D_{\bs+\be_i + \be_k,t} D_{\bs + \be_j,t} + F^{(jk)}_{\bs} D_{\bs+\be_j + \be_k,t} D_{\bs + \be_i,t} = 0, \quad i<j<k.
\end{equation}

The result below implies that, in generic situation, there exists suitable gauge of the wave function $Z_{\bs,t}(x)$ and of the $\tau$-function $D_{\bs,t}$ which allows to go back to the simpler original form of the equations \eqref{eq:D-F-t},  \eqref{eq:Z-D-jk-F} and \eqref{eq:D-ijk-F}.  
\begin{Prop}
For non-vanishing functions $F^{(ij)}_{\bs}$, what we assume in this Section, there exists gauge function $\sigma_{\bs}$ such that the rescaled wave function  and the $\tau$-function 
\begin{equation} \label{eq:Z-D-sigma}
\mathcal{Z}_{\bs,t}(x) = Z_{\bs ,t}(x) \sigma_{\bs}, \qquad \mathcal{D}_{\bs ,t} = D_{\bs , t} \sigma_{\bs},
\end{equation}
satisfy the original versions \eqref{eq:Z-D-jk} and \eqref{eq:D-ijk-H} of the linear problem and the Hirota system in time $t$. Moreover, the modification does not change the linear equations~\eqref{eq:Z-D-j-t}, and brings back generalized equation~\eqref{eq:D-F-t} to its canonical form~\eqref{eq:D-jk-t-H}.	
\end{Prop}
\begin{proof}	
	For $r\geq 3$ we can consider compatibility conditions of the above system~\eqref{eq:Z-D-jk-F}, the counterpart of equations~\eqref{eq:CC-r}, what gives
	\begin{equation} \label{eq:FF}
	F^{(ij)}_{\bs+\be_k} F^{(ik)}_{\bs} = F^{(ik)}_{\bs+\be_j} F^{(ij)}_{\bs}.
	\end{equation}
	For generic non-zero functions $F^{(ij)}_{\bs}$, what we assume in this Section, this implies existence of the potentials $\rho^{(i)}_{\bs}$ such that
	\begin{equation} \label{eq:F-rho}
	F^{(ij)}_{\bs} = \frac{\rho^{(i)}_{\bs + \be_j}}{\rho^{(i)}_{\bs}}.
	\end{equation}
	Moreover, the condition $F^{(ij)}_{\bs} = - F^{(ji)}_{\bs}$ implies existence of a potential $\sigma_{\bs}$ such that
	\begin{equation} \label{eq:rho-sigma}
	\rho^{(i)}_{\bs} = (-1)^{\sum_{j<i} s_j}\frac{\sigma_{\bs+\be_i}}{\sigma_{\bs}}.
	\end{equation}
When $r=2$ then we do not have equations~\eqref{eq:FF}, but we can use equations \eqref{eq:F-rho} and \eqref{eq:rho-sigma} to define the potentials.
By direct calculation one can check that the new wave function $\mathcal{Z}_{\bs,t}(x)$ and new $\tau$-function  $ \mathcal{D}_{\bs ,t}$ have properties described in the statement.	
\end{proof}
\begin{Cor}
	The corresponding versions of equations~\eqref{eq:Z-D-P-j-t} and \eqref{eq:Z-D-P-j-t2} change their forms to
\begin{align}
\mathcal{Z}_{\bs,t+1}(x) \mathcal{D}_{\bs, t} & = \mathcal{Z}_{\bs, t}(x) \mathcal{D}_{\bs, t+1} -  \sum_{i=1}^r \frac{\sigma_{\bs}^2}{\sigma_{\bs+\be_i}\sigma_{\bs-\be_i}} \mathcal{Z}_{\bs - \be_i,t+1}(x) \mathcal{D}_{\bs+\be_i, t},\\
\mathcal{Z}_{\bs,t}(x) \mathcal{D}_{\bs, t} & = \mathcal{Z}_{\bs, t-1}(x) \mathcal{D}_{\bs, t+1} -  x \sum_{i=1}^r \frac{\sigma_{\bs}^2}{\sigma_{\bs+\be_i}\sigma_{\bs-\be_i}} \mathcal{Z}_{\bs - \be_i,t+1}(x) \mathcal{D}_{\bs+\be_i, t-1} ,
\end{align}
respectively. However, when the functions $F^{(ij)}_{\bs}$ are constants
\begin{equation}
F^{(ij)}_{\bs} = F^{(ij)},
\end{equation}
then 
\begin{equation}
\rho^{(i)}_{\bs} = \prod_{j\neq i}\left( F^{(ij)} \right)^{s_j} \qquad \text{and} \qquad \sigma_{\bs} = 
\prod_{i < j}\left( F^{(ij)} \right)^{s_i s_j} , 
\end{equation}
what implies that also the both equations do not change their form in that case.
\end{Cor}

\begin{Prop} \label{prop:ZDF}
	The compatibility of equations \eqref{eq:Z-D-j-t} and \eqref{eq:Z-D-P-j-t2} implies, among others, the following non-autonomous generalization of the multidimensional discrete-time Toda constraint~\eqref{eq:MDTE-D}
\begin{equation} \label{eq:MDTE-D-F}
D_{\bs,t}^2 + \sum_{i=1}^r 	D_{\bs+\be_i, t-1} 	D_{\bs-\be_i, t+1} = F_{|\bs|+t} D_{\bs,t+1} 	D_{\bs,t-1 } ,
\end{equation}		
where $F_{|\bs|+t}$ is a function of the sum $s_1 + \dots + s_r + t$ of discrete space and time variables.
\end{Prop}
\begin{proof}
	The reasoning follows the basic idea applied in \cite{IDS} to verify compatibility of the standard discrete-time Toda equation. It is very similar to the proof of analogous result on Hermite--Pad\'{e} approximation~\cite{Doliwa-Siemaszko-2}. We start from the polynomials $Z_{\bs,t+1}(x)$, $Z_{\bs-\be_j,t+1}(x)$, $j=1, \dots , r$. Using suitably shifted equations \eqref{eq:Z-D-j-t} and \eqref{eq:Z-D-P-j-t2} we obtain polynomials $Z_{\bs, t -1}(x) $ and $Z_{\bs + \be_j, t-1}(x)$
\begin{equation*}
\begin{array}{cclcl}
Z_{\bs,t+1}(x), \; Z_{\bs-\be_j,t+1}(x)  & \Rightarrow &
Z_{\bs - \be_j, t+2}(x) & & \eqref{eq:Z-D-j-t}_{\bs \to \bs- \be_j,t\to t+1}\\
Z_{\bs,t+1}(x), \; Z_{\bs-\be_j,t+2}(x) & \Rightarrow & Z_{\bs, t}(x) & & \eqref{eq:Z-D-P-j-t2}_{t\to t+1} \\
Z_{\bs,t}(x), \; Z_{\bs,t+1}(x)  & \Rightarrow &
Z_{\bs+\be_j, t}(x) & \text{by} & \eqref{eq:Z-D-j-t}\\
Z_{\bs,t}(x), \; Z_{\bs-\be_j,t+1}(x) & \Rightarrow & Z_{\bs, t-1}(x) &  & \eqref{eq:Z-D-P-j-t2} \\
Z_{\bs,t+1}(x), \; Z_{\bs + \be_j,t}(x) \; Z_{\bs - \be_j,t + 1}(x) & \Rightarrow &
Z_{\bs+\be_j, t-1 } & & \eqref{eq:Z-D-P-j-t2}_{\bs \to \bs + \be_j},
\end{array}
\end{equation*} 
where at the last step we use the following consequence of equations~\eqref{eq:Z-D-P-j-t}
\begin{equation*}
Z_{\bs -\be_i + \be_j,t+1}(x) =
Z_{\bs ,t+1}(x) \frac{D_{\bs - \be_i + \be_j,t+1}}{D_{\bs, t+1}} - Z_{\bs -\be_i,t+1}(x)  \frac{D_{\bs - \be_i + \be_j,t+2} D_{\bs, t+1} -  D_{\bs, t+2} D_{\bs - \be_i + \be_j,t+1} }{D_{\bs, t+1}D_{\bs - \be_i,t+2}}.
\end{equation*}
Together with $Z_{\bs, t}(x) $  the polynomials should satisfy equation $\eqref{eq:Z-D-j-t}_{t\to t-1}$, what gives nonlinear equations for $D_{\bs,t}$. 

By collecting terms at  $Z_{\bs-\be_i, t+1}(x) $ in the final equation we obtain 
\begin{equation*}
\left(\frac{D_{\bs+ \be_j,t+1} 	D_{\bs+\be_i,t } - 	D_{\bs+\be_i, t+1} 	D_{\bs+\be_j, t}} {
	D_{\bs+\be_j + \be_i,t} 	D_{\bs,t+1}} \right)_{\bs - \be_i,t+1} =\left(\frac{D_{\bs+ \be_j,t+1} 	D_{\bs+\be_i,t } - 	D_{\bs+\be_i, t+1} 	D_{\bs+\be_j, t}} {
	D_{\bs+\be_j + \be_i,t} 	D_{\bs,t+1}} \right)_{\bs ,t-1} ,
\end{equation*}
what, together with \eqref{eq:D-F-t}, implies that the functions $F^{(ij)}_{\bs}$ are constants $F^{(ij)}$.
The terms at $Z_{\bs, t+1}(x) $  give
\begin{equation*}
\left(\frac{D_{\bs,t}^2 + \sum_{i=1}^r 	D_{\bs+\be_i, t-1} 	D_{\bs-\be_i, t+1}}{D_{\bs,t+1} 	D_{\bs,t-1 }  } \right)_{t \to t+1} =
\left(\frac{D_{\bs,t}^2 + \sum_{i=1}^r 	D_{\bs+\be_i, t-1} 	D_{\bs-\be_i, t+1}}{D_{\bs,t+1} 	D_{\bs,t-1 }  } \right)_{\bs \to \bs+\be_j},
\end{equation*}
for $j=1,\dots ,r$, what concludes the proof.
\end{proof}
\begin{Cor} \label{cor:F}
When the constants $F^{(ij)}$ do not vanish, then by Proposition~\ref{prop:ZDF} and Corollary~\ref{cor:F} without losing generality we can put $F^{(ij)}=1$ for $i<j$.
\end{Cor}

\begin{Com}
	Analogous compatibility condition between equations~\eqref{eq:Z-D-j-t} and \eqref{eq:Z-D-P-j-t} gives also that the functions $F^{(ij)}_{\bs}$ are constants $F^{(ij)}$. Instead of equation \eqref{eq:MDTE-D-F} we have then just $\tau$-function form of \eqref{eq:A+B-r} with the fields $A^{(j)}_{\bs,t}$ and $B^{(j)}_{\bs,t}$  replaced by their definitions \eqref{eq:A-Q-r} and \eqref{eq:B-Q-r}.
\end{Com}
\begin{Rem}
As it was presented in \cite{AptekarevDerevyaginMikiVanAssche} equation~\eqref{eq:A+B-r} can be derived from~\eqref{eq:MDTE-D} by considering the ratio of two such equations with $D_{\bs + \be_j , t+1}$ and $D_{\bs, t+1}$ on their left hand sides. Because of that ratio the reasoning cannot be reversed. 
\end{Rem}
\section{Final remarks}

We presented fundamentals of the general theory of multiple orthogonal polynomials basing on determinantal identities. In looking for unified and symmetric description of the resulting formulas we discovered new equations satisfied by the polynomials. The simple discrete-time evolution of the measures allows to include the multiple orthogonal polynomials into the integrable systems theory, as a reduction of Hirota's discrete Kadomtsev--Petviashvili system in direct analogy with  the Paszkowski constraint of the Hermite--Pad\'{e} approximation theory. 

We are convinced that the multiple orthogonal polynomials will play an increasingly important role in the coming years and expect their range of applications to exceed that of standard orthogonal polynomials, which is already significant.   We believe that many of applications of multiple orthogonal polynomials are still waiting to be discovered and exploited.  Our belief is based on their connection with the theory of integrable systems and the enormous influence that this theory has had on many areas of theoretical physics and applied mathematics over the last fifty years.

%\subsection*{Acknowledgment}

% ------------------------------------------------------------------------

\begin{thebibliography}{spmpsci}
\bibliographystyle{cite}


\bibitem{AdlervanMoerbeke}
Adler, M., Van Moerbeke, P.: Generalized orthogonal polynomials, discrete KP and Riemann-Hilbert problems. Commun. Math. Phys. \textbf{207},  589--620 (1999)

\bibitem{AdlervanMoerbekeVanhaecke}
Adler, M., Van Moerbeke, P., Vanhaecke, P.: Moment matrices and multi-component KP,
	with applications to random matrix theory.
Commun. Math. Phys. \textbf{286},  1--38 (2009)

\bibitem{Akhiezer}
Akhiezer, N.I.: The classical moment problem and some related questions in analysis. Olivier \& Boyd, Edinburgh and London (1963)

\bibitem{Alvarez-FernandezPrietoManas}
\'{A}lvarez-Fern\'{a}ndez, C., Fidalgo Prieto, U., Ma\~{n}as, M.: \emph{Multiple orthogonal polynomials of mixed type:	Gauss--Borel factorization and the multi-component
	2D Toda hierarchy}, Adv. Math. \textbf{227} (2011) 1451--1525. 

\bibitem{Aptekarev}
Aptekarev, A.I.: Multiple orthogonal polynomials. J. Comput. Appl. Math. \textbf{99},  423--447 (1998)

\bibitem{AptekarevBranquinhoVanAssche}
Aptekarev, A.I., Branquinho, A., Van Assche, W.: Multiple orthogonal polynomials for classical weights. Trans. AMS \textbf{355}, 3887--3914 (2003)

\bibitem{AptekarevKuijlaars}
Aptekarev, A.I., Kuijlaars, A.: Hermite--Pad\'{e} approximations and multiple orthogonal polynomial ensembles, Russian Mathematical Surveys \textbf{66}, 1133--1199 (2011)


\bibitem{AptekarevDerevyaginMikiVanAssche}
Aptekarev, A.I., Derevyagin, M., Miki, H., Van Assche, W.: Multidimensional Toda lattices: continuous and discrete time. SIGMA \textbf{12}, 054 (2016)


\bibitem{Baker}
Baker, Jr., G.A., Graves-Morris, P.: Pad\'{e} approximants. Cambridge University Press, Cambridge (1996)

\bibitem{BeckermannLabahn}
Beckermann, B., Labahn, G.: Fraction-free computation of simultaneous Pad\'{e} approximants. In: Proceedings of the 2009 International Symposium on Symbolic and Algebraic Computation ISSAC'09, pp. 15--22. ACM, New York (2009)

\bibitem{BialeckiDoliwa}
Bia{\l}ecki, M., Doliwa, A.: Algebro-geometric solution of the discrete KP equation over a finite field out of a hyperelliptic curve. Commun. Math. Phys. \textbf{253}, 157--170 (2005) 

\bibitem{BleherKuijlaars}
Bleher, P.M., Kuijlaars, A.B.J.: Random matrices with external source and multiply orthogonal polynomials. Int. Math. Res. Not. \textbf{2004}, 109--129 (2004)


\bibitem{BranquinhoMorenoManas}
Branquinho, A., Moreno, A.F., Ma\~{n}as, M.: Multiple orthogonal polynomials: Pearson equations and Christoffel formulas. Analysis \& Math. Phys. \textbf{12}, 129 (2022)

\bibitem{Brezinski}
Brezinski, C.: History of continued fractions and Pad\'{e} approximants. Springer, Berlin (1991)

\bibitem{BrualdiSchneider}
Brualdi, R., Schneider, H.: Determinantal identities: Gauss, Schur, Cauchy, Sylvester, Kronecker, Jacobi, Binet, Laplace, Muir, and Cayley. Lin. Alg. Appl. \textbf{52-53}, 769--791 (1983)

\bibitem{CMGV-2}
Cantero, M.J., Moral, L., Gr\"{u}nbaum, F., Vel\'{a}zquez, L.: The CGMV method for quantum walks.
Quantum Information Processing \textbf{11}, 1149--1192 (2012)

\bibitem{Chihara}
Chihara, T.S.: An Introduction to Orthogonal Polynomials. Gordon and Breach, New York (1978)

\bibitem{Clarkson}
Clarkson, P.A.: Recurrence coefficients for discrete orthonormal polynomials and the Painlev\'{e} equations. J.~Phys. A: Math. Theor. \textbf{46}, 185205 (2013)

\bibitem{CoussementVanAasche}
Coussement, J., Van Assche, W.: Gaussian quadrature for multiple orthogonal polynomials. J. Comput. Appl. Math. \textbf{178}, 131--145 (2005)

\bibitem{DaemsKuilaars}
Daems, E., Kuijlaars, A.B.J.: Multiple orthogonal polynomials of mixed type and non-intersecting Brownian motions. J. Approx. Theory \textbf{146}, 91--114 (2007)


\bibitem{DKJM}
Date, E., Kashiwara, M., Jimbo, M., Miwa, T.: Transformation groups for
	soliton equations. In: Jimbo, M.,  T. Miwa, T.  (eds.) Nonlinear integrable systems --- classical theory and
quantum theory, Proc. of RIMS Symposium, pp. 39--119. World
Scientific, Singapore (1983)
	
\bibitem{Deift}
Deift, P.A.: Orthogonal Polynomials and Random Matrices: A Riemann-Hilbert Approach. AMS, New York (1998)


\bibitem{Iglesia}
de la Iglesia, M.D.: Orthogonal Polynomials in the Spectral Analysis of Markov Processes. Cambridge University Press, Cambridge (2022)	

\bibitem{Dol-Des} 
Doliwa, A.: Desargues maps and the Hirota--Miwa equation. Proc. R. Soc. A
\textbf{466},  1177--1200 (2010)

\bibitem{Dol-AN} 
Doliwa, A.: The affine Weyl group symmetry of Desargues maps and of the non-commutative Hirota--Miwa system. Phys. Lett. A {\bf 375},  1219--1224 (2011)


\bibitem{Doliwa-NHP}
Doliwa, A.: Non-commutative Hermite--Pad\'{e} approximation and integrability. Lett. Math. Phys. \textbf{112}, 68 (2022)

\bibitem{Doliwa-NAMTS}
Doliwa, A.: Non-autonomous multidimensional Toda system and multiple interpolation problem. J. Phys. A: Math. Theor. \textbf{55}, 505202  (2022) 

\bibitem{Doliwa-HP-MP-T}
Doliwa, A.: Hermite-Pad\'{e} approximation, multiple orthogonal polynomials, and multidimensional Toda equations, \texttt{arXiv:2310.15116}

\bibitem{Doliwa-Siemaszko-2}
Doliwa, A., Siemaszko, A.: Hermite--Pad\'{e} approximation and integrability. J. Approx. Theory \textbf{292}, 105910 (2023)

\bibitem{DoliwaSiemaszko-QW}
Doliwa, A., Siemaszko, A.: Spectral quantization of discrete random walks on half-line, and orthogonal polynomials on the unit circle. \texttt{arXiv:2306.12265}


\bibitem{FernandezManas}
Fern\'{a}ndez-Irrisarri, I., Ma\~{n}as. M.,: Toda and Laguerre-Freud equations and tau functions for hypergeometric discrete multiple orthogonal polynomials. \texttt{arXiv:2307.08075}


\bibitem{FilipukVanAsscheZhang}
Filipuk, G., Van Assche, W., Zhang, L.:
Ladder operators and differential equations for multiple orthogonal polynomials. J. Phys. A: Math. Theor. \textbf{46}, 205204 (2013) 

\bibitem{Flajolet}
Flajolet, P.: Combinatorial aspects of continued fractions. Discrete Math. \textbf{32}, 125--162 (1980)

\bibitem{Flaschka}
Flaschka, H.: The Toda lattice. I.
Existence of integrals, Phys. Rev. B  \textbf{9}, 1924--1925 (1974)

\bibitem{Geronimus}
Geronimus,	Ya.: Polynomials orthogonal on a circle and interval. Pergamon Press, Oxford (1960)


\bibitem{Gragg}
Gragg, W.B.: The Pad\'{e} table and its relation to certain algorithms of numerical analysis. SIAM Review \textbf{14}, 1--62 (1972)

\bibitem{Harnad}
Harnad, J., Balogh, F.: Tau Functions and Their Applications, Cambridge University Press, Cambridge (2021)

\bibitem{IDS}
Hietarinta, J., Joshi, N., Nijhoff, F.W.: Discrete systems and integrability. Cambridge University Press, Cambridge (2016)

\bibitem{Hirota-2dT}
Hirota, R.: Nonlinear partial difference equations. II. Discrete-time Toda equation. J. Phys. Soc. Japan, \textbf{43},  2074--2078 (1977)

\bibitem{Hirota}
Hirota, R.: Discrete analogue of a generalized Toda equation. J. Phys. Soc. Japan \textbf{50}, 3785--3791 (1981)


\bibitem{Hirota-book}
Hirota, R.: The direct method in soliton theory.
Cambridge University Press, Cambridge (2004)

\bibitem{Hirota-1993}
Hirota, R., Tsujimoto, S., Imai, T.: Difference scheme of soliton equations. In: Christiansen, P.L., Eilbeck, P.L., Parmentier, R.D. (eds.) Future Directions of Nonlinear Dynamics in Physical and Biological Systems, pp.~7--15. Springer (1993)

\bibitem{Ismail}
Ismail, M.E.H.: Classical and Quantum Orthogonal Polynomials in One Variable.  Cambridge University Press, Cambridge (2005)

\bibitem{KarlinMcGregor}
Karlin, S., McGregor, J.L.: The differential equations of birth-and-death processes, and the
Stieltjes moment problem. Trans. Amer. Math. Soc. 
 \textbf{85}, 489--546 (1957)

\bibitem{Klimyk-Schmudgen}
Klimyk, A., Schm\"{u}dgen, K.: Quantum Groups and Their Representations. Springer, Berlin (1997)

\bibitem{Krichever-alg-geom-discrete}
Krichever, I.M.: Algebraic curves and non-linear difference equations. Russian Math. Surv. \textbf{33}, 255-256 (1978)

\bibitem{Kuijlaars}
Kuijlaars, A.B.J.: Multiple orthogonal polynomial ensembles. Recent trends in orthogonal polynomials and approximation theory. Contemp. Math. \textbf{507}, 155--176 (2010) 


\bibitem{KNS-rev}
Kuniba, A., Nakanishi, T., Suzuki, J.: $T$-systems and $Y$-systems in integrable systems. 
J. Phys. A: Math. Theor. {\bf 44}, 103001 (2011) 


\bibitem{LedermannReuter}
Ledermann, W., Reuter, G.E.H.: Spectral theory for the differential equations of simple
birth and death processes. Philos. Trans. Roy. Soc. London A \textbf{246}, 321--369 (1954)


\bibitem{LeBen}
Levi, D., Benguria, R.: B\"{a}cklund transformations and nonlinear
	differential-difference equations. Proc. Nat. Acad. Sci. USA \textbf{77}, 5025--5027 (1980)


\bibitem{Mahler-P}
Mahler, K.: Perfect systems. Compositio Math. \textbf{19},  95--166 (1968)

\bibitem{MF-VA}
Mart\'{i}nez-Finkelshtein, A., Van Assche, W.:  WHAT IS...A Multiple Orthogonal Polynomial? Notices AMS \textbf{63}, 1029--1031 (2016)

\bibitem{Miwa} 
Miwa, T.: On Hirota's difference equations. 
Proc. Japan Acad. \textbf{58}, 9--12 (1982) 

\bibitem{Moser}
Moser, J.: Three integrable Hamiltonian systems connected with isospectral deformations. Adv. Math. \textbf{16}, 197--220 (1975) (1975)

\bibitem{Nagao}
Nagao, H.: The Pad\'{e} interpolation method applied to additive difference Painlevé equations. Lett. Math. Phys.  \textbf{111}, 135 (2011)

\bibitem{Nagao-Yamada}
Nagao, H., Yamada, Y.: Pad\'{e} Methods for Painlev\'{e} Equations. Springer, Singapore (2021)

\bibitem{Nevai}
Nevai, P. (ed).: Orthogonal Polynomials:
Theory and Practice. Kluwer Academic Publishers, Dordrecht--Boston--London (1990) 

\bibitem{Nikiforov-Suslov-Uvarov}
Nikiforov, A.F., Suslov,  S.K., Uvarov,  V.B.: Classical Orthogonal Polynomials of a Discrete Variable.
Springer-Verlag, Berlin-Heidelberg (1991)


\bibitem{NikishinSorokin}
Nikishin, E.M.,  Sorokin, V.N.: Rational approximation and orthogonality. Transl. Math. Monographs \textbf{92}, Amer. Math. Soc. (1991)

\bibitem{Nimmo-NCKP}
Nimmo, J.J.C.: On a non-Abelian Hirota-Miwa equation.
J. Phys. A: Math. Gen. \textbf{39}, 5053--5065 (2006) 

\bibitem{Paszkowski}
Paszkowski, S.: Recurrence relations in Pad\'{e}--Hermite approximation. J. Comput.
Appl. Math. \textbf{19}, 99--107  (1987) 


\bibitem{Saito-Saitoh}
Saito, S., Saitoh, N.: Gauge and dual symmetries and linearization of 	Hirota's bilinear equations. J.~Math. Phys. \textbf{28}, 1052--1055 (1987)

\bibitem{Sato}
Sato, M.: Soliton Equations as Dynamical Systems on a Infinite
Dimensional Grassmann Manifolds. RIMS, Kokyuroku, Kyoto Univ. \textbf{439}, 30--46 (1981)

\bibitem{Schoutens}	
Schoutens, W.: Stochastic Processes and Orthogonal Polynomials. Springer, New York (2000)

\bibitem{Shiota}
Shiota, T.: Characterization of Jacobian varieties in terms of soliton equations. Invent. Math. \textbf{83}, 333--382  (1986)

\bibitem{Sogo}
Sogo, K.: Time-dependent orthogonal polynomials and theory of solitons --- applications to matrix model, vertex model and level statistics. J. Phys. Soc. Japan \textbf{62}, 1887--1894 (1993)


\bibitem{Szego}
Szeg\H{o}, G.: Orthogonal polynomials. AMS, Providence, RI, fourth edition (1975)


\bibitem{Toda-TL}
Toda, M.: Waves in nonlinear lattice.
Progr. Theoret. Phys. Suppl. \textbf{45}, 174--200 (1970)

\bibitem{VanAsche-HPAO}
Van Assche, W.: Pad\'{e} and Hermite--Pad\'{e} approximation and orthogonality. Surv. Approx. Theory \textbf{2}, 61--91 (2006)

\bibitem{VanAsche-nn-mop}
Van Assche, W.: Nearest neighbor recurrence relations for multiple orthogonal polynomials. J. Approx. Theory \textbf{163}, 1427--1448 (2011)

\bibitem{VanAssche}
Van Assche, W.: Orthogonal polynomials and Painlev\'{e} equations. Cambridge University Press (2018)

\bibitem{VanAsscheCoussement}
Van Assche, W., Coussement, E.: Some classical multiple orthogonal polynomials. J. Comput. Appl. Math. \textbf{127}, 317--347 (2001) 

\bibitem{Moerbeke}
Van Moerbeke, P.: The Spectrum of Jacobi Matrices. Inventiones Math. \textbf{37}, 45--81 (1976) 

\bibitem{Viennot}
Viennot, G.: A Combinatorial Theory for General Orthogonal Polynomials with Extensions and Applications. Orthogonal Polynomials and Applicationsv (Bar-le Duc, 1984), pp. 139--157. Lecture Notes in Math. \textbf{1171}, Springer, Berlin (1985) 

\bibitem{Vilenkin-Klimyk}
Vilenkin, N.Ja., Klimyk,  A.U.: Representation of Lie Groups and Special Functions: Volume 1: Simplest Lie Groups, Special Functions and Integral Transforms, Volume 2: Class I Representations, Special Functions, and Integral Transforms, Volume 3: Classical and Quantum Groups and Special Functions. Kluwer Academic Publishers (1991, 1993, 1992)

\bibitem{WahlquistEstabrook} 
Wahlquist, H.D., Estabrook, F.B.: B\"{a}cklund transformation for solutions of the Korteweg--de~Vries equation. Phys. Rev. Lett. \textbf{31}, 1386--1390 (1973)


\bibitem{Yamada}
Yamada, Y.: Pad\'{e} method to Painlev\'{e} equations. Funkcialaj Ekvacioj \textbf{52}, 83--92 (2009) 

\bibitem{Zabrodin}
Zabrodin, A.V.: Hirota’s difference equations. Theor. Math. Phys. \textbf{113}, 1347--1392 (1997)

\end{thebibliography}
\end{document}